\documentclass[12pt]{article}
\usepackage{amsmath}
\usepackage{graphicx,psfrag,epsf}
\usepackage{enumerate}
\usepackage{url} 
\usepackage{multibib}
\newcites{supp}{Supplementary References}

\usepackage[a4paper,
	    left=0.65in,
            right=0.7in,
            top=1in,
            bottom=1in]{geometry}

\usepackage{blindtext}

\usepackage{graphicx, url}
\usepackage{amsmath}
\usepackage{dsfont}
\usepackage{amssymb}
\usepackage{epstopdf}
\usepackage{boxedminipage}
\usepackage{amsfonts}
\usepackage{url}
\usepackage{amsmath}
\usepackage{setspace}
\usepackage{natbib}
\usepackage{bm}
\usepackage{wasysym}
\usepackage{textcomp}
\usepackage[export]{adjustbox}
\usepackage{mathtools} 
\usepackage[normalem]{ulem}
\usepackage{setspace,lscape}
\usepackage{tabularx,booktabs,multirow}
\usepackage[usenames]{color}

\usepackage{color,latexsym,epsfig,graphics,graphicx,amssymb,amsbsy,amsmath,amsthm}
\usepackage{rotating}
\usepackage[arrow,curve,all]{xy}
\usepackage{setspace,lscape}
\usepackage{tabularx,booktabs,multirow}
\usepackage{tikz}
\usetikzlibrary{shapes.callouts}
\usetikzlibrary{snakes}

\newcommand{\blue}{\color{black}}

\newtheorem{theorem}{Theorem}
\newtheorem{corollary}{Corollary}
\newtheorem{proposition}{Proposition}
\newtheorem{lemma}{Lemma}
\newtheorem{example}{Example}
\newtheorem{definition}{Definition}
\newtheorem{remark}{Remark}
\newtheorem{proc}{Procedure}

\newcommand{\beq}{\begin{equation}}
\newcommand{\eeq}{\end{equation}}
\newcommand{\beas}{\begin{eqnarray*}}
\newcommand{\eeas}{\end{eqnarray*}}
\newcommand{\bea}{\begin{eqnarray}}
\newcommand{\eea}{\end{eqnarray}}
\newcommand{\bei}{\begin{itemize}}
\newcommand{\eei}{\end{itemize}}
\newcommand{\ben}{\begin{enumerate}}
\newcommand{\een}{\end{enumerate}}
\newcommand{\bet}{\begin{theorem}}
\newcommand{\eet}{\end{theorem}}
\newcommand{\bel}{\begin{lemma}}
\newcommand{\eel}{\end{lemma}}
\newcommand{\bep}{\begin{proposition}}
\newcommand{\eep}{\end{proposition}}
\newcommand{\bed}{\begin{definition}}
\newcommand{\eed}{\end{definition}}
\newcommand{\bec}{\begin{corollary}}
\newcommand{\eec}{\end{corollary}}
\newcommand{\bex}{\begin{example}}
\newcommand{\eex}{\end{example}}

\newcommand{\PP}{\mathbb P}
\newcommand{\EE}{\mathbb E}

\newcommand{\PBD}{\texttt {PBD}}

\def\0{\boldsymbol{0}}

\def\Z{\boldsymbol{Z}}
\def\z{\boldsymbol{z}}

\def\H{\mathcal{H}}

\newcommand{\lfdr}{\emph{lfdr}}

\newenvironment{boxedtable}[2]{\begin{table}[!htbp]\noindent{\bf \caption{#1 \label{#2}}}\small} {\noindent\end{table}}

\newenvironment{boxedfigure}[2]{\begin{figure}[!htbp]\noindent{\bf \caption{#1 \label{#2}}}\small} {\noindent\end{figure}}

\newcommand{\tablexplain}[1]{\noindent\small{#1}}
\newenvironment{ctabular}{\begin{center}\begin{tabular}}{\end{tabular}\end{center}}

\newcolumntype{Y}{>{\centering\arraybackslash}X}

\begin{document}

\def\spacingset#1{\renewcommand{\baselinestretch}%
{#1}\small\normalsize} \spacingset{1}


\title{\Large \textbf{An Empirical Bayes Approach to Controlling the False Discovery Exceedance}}


\author{
\begin{tabular}[t]{c@{\extracolsep{4em}}c} 
\textbf{Pallavi Basu}  & \textbf{Luella Fu} \\
{\normalsize Indian School of Business} & {\normalsize San Francisco State University}\\
{\normalsize Statistics \textit{\&} OR, Tel Aviv University}& {\normalsize \it{luella@sfsu.edu}}\\
{\normalsize \it{pallavi\_basu@isb.edu}}\\
\\
\textbf{Alessio Saretto}  & \textbf{Wenguang Sun} \\
{\normalsize Federal Reserve Bank of Dallas} & {\normalsize USC Marshall School of Business}\\
{\normalsize \it{alessio.saretto@dal.frb.org}}& {\normalsize \it{wenguans@marshall.usc.edu}}\\
\end{tabular}
}

\date{}


\maketitle

\abstract{\noindent In large-scale multiple hypothesis testing problems, the false discovery exceedance (FDX) provides a desirable alternative to the widely used false discovery rate (FDR) when the false discovery proportion  (FDP) is highly variable. We develop an empirical Bayes approach to control the FDX. We show that, for independent hypotheses from a two-group model and dependent hypotheses from a Gaussian model fulfilling the exchangeability condition, an oracle decision rule based on ranking and thresholding the local false discovery rate (\lfdr) is optimal in the sense that the power is maximized subject to the FDX constraint. We propose a data-driven FDX procedure that uses carefully designed computational shortcuts to emulate the oracle rule. We investigate the empirical performance of the proposed method using both simulated and real data and study the merits of FDX control through an application for identifying abnormal stock trading strategies.


\vspace{1in}
\begin{quote}
\end{quote}



\spacingset{1.8}

\section{Introduction}


\subsection{False Discovery Proportion}

Multiple hypothesis testing provides a useful and powerful technique for identifying sparse signals in massive data. To account for the multiplicity in large-scale testing problems, the false discovery rate (FDR; \citealp{BenHoc95}) has been widely used, and a plethora of FDR procedures, exemplified by the Benjamini-Hochberg (BH) procedure, have been developed to control the FDR. Control of the FDR guarantees that the expected value of the false discovery proportion (FDP) over repeated trials is below a pre-specified level. However, the simple and elegant FDR notion still allows for 
large variability in the FDPs. 
Specifically, the effective control of the FDR at the nominal level $\alpha$ does not imply the effective control of the FDP in a particular experiment. This is concerning since scientific findings are often reported based on experiments performed only a few times or even just once, exacerbating the replicability crisis. The issue can be particularly worrying in applications where the signals are sparse and weak, or where the tests are highly dependent -- in such settings classical methods such as BH often exhibit unstable behaviors, leading to highly variable and possibly skewed FDPs across different experiments \citep{Kornetal2004, DelattreRoquain2015}. As an extreme example, consider a multiple testing procedure that produces FDPs of $0.2$ in $50\%$ of the trials and makes no rejection in the other trials (by convention the FDP in these trials is set $0$). Suppose that the nominal FDR level is 0.1. Then, this procedure successfully produces an FDR at exactly the nominal level, but it does a poor job in every single trial, where in half the trials, the procedure produces too many false rejections and in the other half, the procedure fails to find any signals. This is clearly an undesirable situation, where more stringent FDP control is needed. 


\subsection{False Discovery Exceedance}



We propose to directly control the tail probability of the FDP, also known as the \emph{false discovery exceedance} (FDX). The FDX provides a more appropriate alternative to the FDR in practical situations where the FDPs are highly variable. Concretely, an FDX procedure at level $(\gamma, \alpha)$ aims to keep the probability that the FDP exceeds a tolerable proportion of false rejections, $\gamma$, at less than a pre-specified small number, $\alpha$. For example, an FDX procedure at level $(10\%, 5\%)$ guarantees that the probability of having 10\% or more false discoveries among all rejections is less than 5\%.  In contrast with the FDR that has no control of large FDPs in specific trials, the FDX criterion guarantees that a \emph{single} implementation of an FDX method has at least a 95\% chance of producing no more than 10\% false rejections.

The FDX can be understood via the notion of $k$--family-wise error rate ($k$-FWER), which generalizes the FWER by incorporating a tolerable number of false rejections, $k$, into its definition. Both the $k$-FWER and FDX procedures aim to control the probability of an undesirable event to be small, with the key difference being whether the event is characterized by a count (FWER) or a proportion (FDX) of false positives. The merits of controlling the FDX have been discussed in \citet{GenoWasser2004,GenoWasser2006, GuoRomano07, ChiTan08, gordon2008optimality, DelattreRoquain2015}, among others. A clear advantage for the characterization of the tolerance level via a proportion is that the FDX methods can easily scale up with the number of tests, as FDR methods do. Various FDX procedures, some of which will be discussed and compared 
in later sections, have been proposed in the literature \citep{LehmannRomano05, ChiTan08, RoquainVillers11, dohler2020controlling}. 


\subsection{Empirical Bayes (EB) FDX Control}

We propose an empirical Bayes procedure for FDX control that relies on first ranking the local false discovery rates (\lfdr) and second using the Poisson binomial distribution to compute the \lfdr \ cutoff. By contrast, most existing FDX procedures are based on $p$-values. The hypotheses are first ordered according to their respective $p$-values and then a $p$-value cutoff is determined based on the pre-specified FDX level $(\gamma, \alpha)$. The choice of cutoff is the key difference that distinguishes existing $p$-value methods from one another.
We prove that the \lfdr \ ranking is optimal \citep[see also,][]{fu2018nonparametric} and that the procedure controls the FDX at the pre-specified level \citep[see also,][]{basu2016model}. Useful techniques are developed such as the Poisson binomial characterization of the false discovery process and efficient computational cutoffs. We demonstrate that the new techniques lead to substantial power gain over conventional $p$-value based methods. The proposed FDX procedure is simple and fast to implement, and is capable of handling millions of tests. We give an overview of the algorithm in the prototype below. 

\begin{center}
\fbox{\begin{minipage}{38em}
\textbf{The proposed EB-FDX procedure:~A prototype algorithm} 
\begin{description}
 \item Step 1. Compute the \emph{lfdr}, adjusting for empirical null if needed. Sort all the hypotheses by increasing order of the estimated \emph{lfdr} statistics. 
 
 \item Step 2. Compute the probability of the undesirable event that the cumulative failure proportion is greater than $\gamma$ using Poisson binomial distribution (PBD); the Bernoulli probabilities in the PBD are the lfdr values as computed in Step 1. 
 
 \item Step 3. Determine the maximum number of rejections via carefully constructed computational shortcuts (Section \ref{sec:computational-shortcuts}) to ensure that the probability of the undesirable event is smaller than or equal to $\alpha$.
 
\end{description}

\end{minipage}}
\end{center}
\vspace{0.2in}

\subsection{Our Contributions}


Under the widely used two-group model (\citealp{efron2001empirical}), we formulate the FDX control problem as a constrained optimization problem where the goal is to maximize the expected number of true positives subject to the FDX constraint characterized by the pair $(\gamma, \alpha)$. Our work makes several contributions to theory, methodology,  and practice. 


\begin{itemize}

\item We propose a new empirical Bayes approach to FDX control based on \lfdr \ and illustrate its efficiency gain over existing frequentists $p$-value methods such as \cite{LehmannRomano05} and \cite{GuoRomano07}. 

\item We establish an optimality theory for FDX control by showing that the \emph{lfdr} ranking is optimal in the sense that the thresholding rule based on \emph{lfdr} has the largest power among all FDX procedures that fulfill the FDX constraint. Although there is a large body of work on FDX, the important optimality issue has not previously been well understood. 

\item We develop an efficient computational algorithm for determining a data-driven cutoff along the \emph{lfdr} ranking; this greatly reduces the computational complexity and enables the broad applicability of our method to large-scale problems with millions of tests. We provide theoretical support to justify the shortcut and illustrate that it significantly improves the computational efficiency using numerical examples.

\item We demonstrate the strong empirical performance of the proposed method via both simulated and real data sets. The key strengths of our methodology include simplicity, power gain, and fast computation. 

\end{itemize}

\subsection{Connection to Existing Work}

The proposed procedure builds upon and contributes to the FDX literature. The very notion of controlling the \textit{exceedance} of the FDP was first defined in \cite{GenoWasser2006}. 
The authors constructed a confidence envelope for FDP by inverting a set of uniformity tests. The approach can handle correlated tests from random fields \citep{perone2004false}. The other approach to the FDX problem is through augmentation, which involves expanding the rejection region from FWER procedures \citep{van2004augmentation, farcomeni2009generalized}. \cite{van2005empirical} proposed using  bootstrap-based Monte Carlo methods to generate the states of hypotheses conditional on observed data, and showed that the method is  more powerful than inversion and augmentation methods. However, all above mentioned methods involve ranking hypotheses using $p$-values or adjusted $p$-values. For earlier works on this topic, see \cite{farcomeni2008review} for a detailed review. A bootstrap-based heuristics is developed in \cite{romano2007control}, later formally justified by \cite{DelattreRoquain2015}.


A more recent and closely related line of research aims to provide  confidence bounds for the FDP with respect to a user-specified (aka post hoc) rejection sets, see, for example, \cite{hemerik2019permutation}. Some follow-up works to \cite{hemerik2019permutation} include \cite{blanchard2020post}, \cite{katsevich2020simultaneous} and \cite{goeman2021only}, among others. However, the efficient ranking of the hypotheses is not discussed, and the tightness of the FDP bounds/thresholds remains unknown in these works. This article addresses both issues by developing a new FDX procedure which (a) ranks hypotheses using the \lfdr \ and (b) determines the \lfdr \ cutoff using an efficient and powerful data-driven algorithm. 

Our work is closely related to the method in \cite{dohler2020controlling}, where the Poisson binomial distribution is employed to determine the $p$-value threshold under the frequentist setting. From a theoretical standpoint, the derivation of our FDX procedure resembles the techniques in \citet{heller2021optimal} for optimal multiple testing under an empirical Bayes setting preceded by the thesis works in \citep{basu2016model} and \citep{fu2018nonparametric} on topics in weighted FDR and FDX control. However, the optimal ranking under the FDX formulation, the connection between Poisson binomial distribution and \lfdr, and the empirical Bayes approach to FDX control are new to the literature. Finally, we demonstrate the power gain over other FDX procedures in simulations and provide a real-life application for uncovering interesting financial trading strategies.  

\subsection{Example and Illustrations}

We provide in Figure \ref{Fig.FDX_demo} an illustration of the proposed FDX controlling procedure and highlight a few differences between FDR and FDX control. We compare two \lfdr--based methods:~the adaptive $z$-value procedure, which controls FDR at level 0.1 (\citealp{SunCai07}), and the proposed procedure which controls FDX at level 0.1 and 0.05, (i.e., the probability is kept below 0.05 that the FDP will exceed 0.1). The methods are performed for $5,\!000$ trials. In each trial, because both methods use the \lfdr, the rankings of the hypotheses from the two methods 
are identical, but the different false discovery control mechanisms create different cutoffs along the \lfdr \ ranking. Figure \ref{Fig.FDX_demo} contrasts the $5,\!000$ FDPs resulting from both methods. In contrast with the FDR procedure (left) that only controls the average of the FDPs over many experiments, resulting in roughly half of trials exceeding the target level (black line), the FDX procedure controls the probability of obtaining a high FDP for any trial, resulting in only a small proportion of trials exceeding the target level. Consequently,  the FDP distribution in the proposed FDX procedure is shifted to the left compared to that of the FDR procedure. This illustrates one important merit of the FDX criterion:~FDX methods can effectively reduce the chance of having high FDPs, which may increase the replicability of scientific discoveries. 

\begin{boxedfigure}{Contrasting FDR and FDX methods}{Fig.FDX_demo}
\tablexplain{Realized FDPs of $5,\!000$ replications for $5,\!000$ tests from the following Gaussian mixture model: $0.8 \ \mathcal{N}(0, 1) + 0.2 \ \mathcal{N}(-2, 1)$. This illustration contrasts our proposed FDX procedure with the lfdr-based FDR procedure proposed by \cite{SunCai07}. Both procedures use the oracle version of the \emph{lfdr} test statistic for demonstration purposes. The \cite{SunCai07} procedure aims to control FDR at 0.1, while the FDX procedure aims to keep the probability that FDP exceeds 0.1 below 0.05.}
\begin{center}
\includegraphics[width=0.8\textwidth]{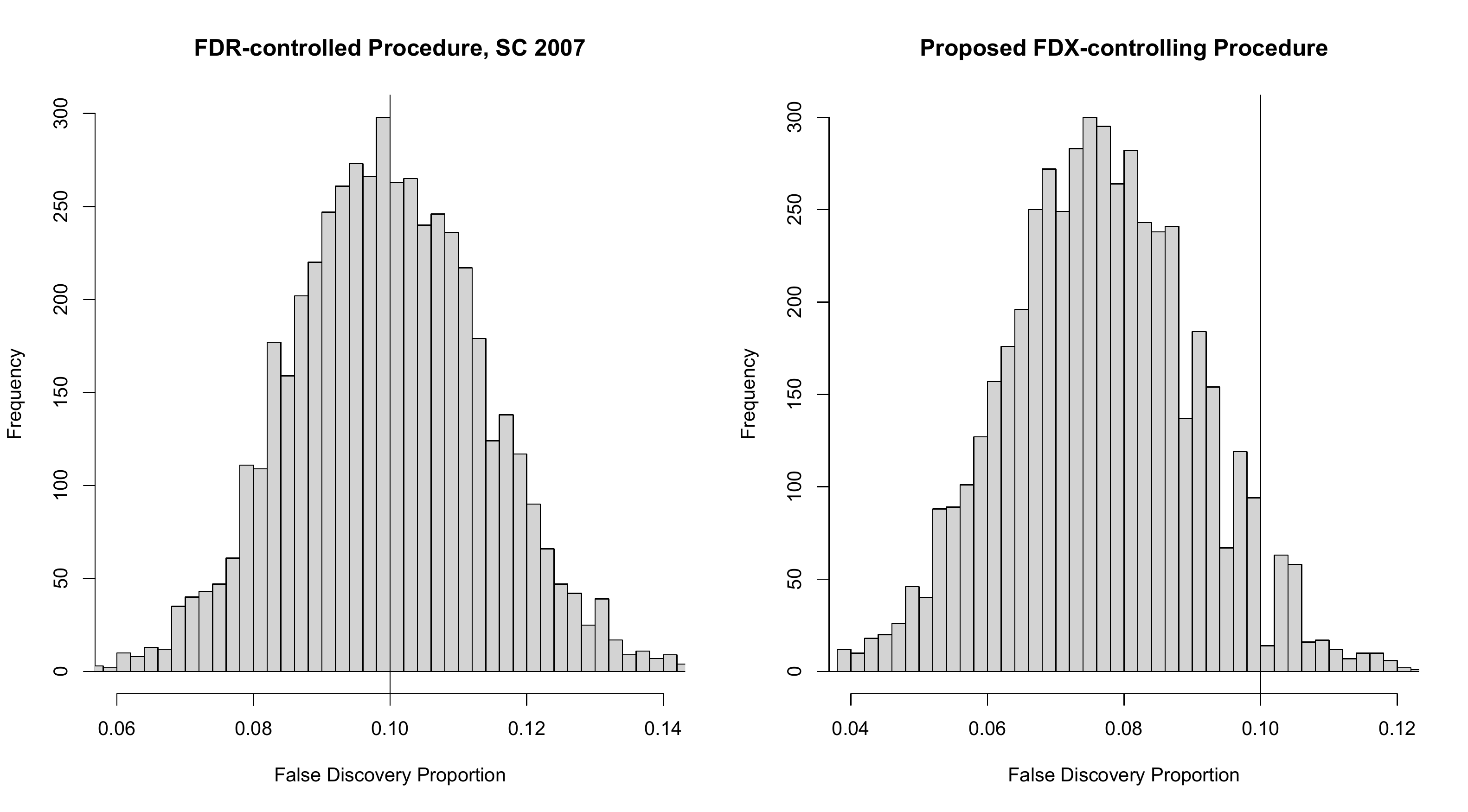}
\end{center}
\end{boxedfigure}


\subsection{Organization}

The rest of the paper is organized as follows. Section 2 introduces notation and sets up the problem. Section 3 describes our proposed solutions. Section 4 discusses further practical concerns about the implementation and \lfdr, which give a greater understanding of the proposed methodology. Section 5 provides numerical simulations. {\blue Section 6 provides an application to real data.} Section 7 concludes with some future propositions. All code for the procedure and experiments can be requested from the first author. Selected proofs and further numerical experiments are described in the supplementary material.


\section{Problem Formulation}\label{sec:formulation}
\subsection{Model and Notation}\label{sec:notation}

Our analysis is based on the premise of a widely used two-group model first described by \cite{efron2001empirical}. Suppose we are interested in testing $m$ hypotheses: $\H = (H^0_i, H^1_i)_{i \in [m]}$, where $i$ represents the index of a study unit; $H^0_i: \theta_i = 0$ and $H^1_i: \theta_i = 1$ are respectively the null and alternative hypotheses corresponding to study unit $i$; $[m]\coloneqq \{1, \hdots, m\}$ denotes the index set of all hypotheses; $\theta_i\in\{0, 1\}$ is the true state of nature with $\theta_i = 0$ representing a true null hypothesis and $\theta_i = 1$ a true alternative. We assume that $\theta_i$ are independent and identically distributed (i.i.d.) variables, obeying a Bernoulli distribution with success probability $\pi=\PP(\theta_i=1)$. The observed data are summarized as a vector of $z$-values $\Z = (Z_i)_{i \in [m]}$, whose distribution can be described using the following (hierarchical) two-group model:
\begin{equation} \label{eq:mixmodel}
\begin{aligned}
& \theta_i \stackrel{i.i.d.}{\sim} \mbox{Bernoulli} (\pi), \\
& Z_i|\theta_i \sim (1-\theta_i) F_{0} + \theta_i  F_{1}, \quad i\in [m],
\end{aligned}
\end{equation}
where $F_{0}$ and $F_1$ are the null and alternative distributions, respectively. Let 
\begin{equation}\label{mixF}
F=(1-\pi)F_0+\pi F_1
\end{equation}
be the mixture distribution. Let $f_0$, $f_1$, and $f$ be the corresponding densities, {\blue which are assumed to exist}. $F_0$ is assumed to be known and taken as the cumulative distribution function (CDF) of an $\mathcal N(0, 1)$ variable. However, \cite{efron2004large} and \cite{JinCai07} argued that the \emph{empirical null} should be employed in practice and provided methods for estimating the {empirical null} distribution. For the present, we assume that $\pi$, $f_0$ and $f$ are known. We discuss the estimation of $\pi$, $f_0$ and $f$  and investigate issues related to the empirical null in Sections \ref{sec:Implementation} and \ref{Sec.Application}. 

A multiple testing rule, which involves making $m$ simultaneous decisions,  can be represented by a binary vector $\pmb\delta=(\delta_i)_{i\in[m]} \in \{0, 1\}^m$. The rule is based on observed data and reflects our belief about the unknown $\pmb\theta=(\theta_i)_{i\in[m]}$. The decision $\delta_i = 1$ indicates the rejection of the null hypothesis (aka ``a statistical discovery"), whereas $\delta_i = 0$ indicates failure to reject the null. The data produces $z$-values, which can then be converted to a significance index that indicates the strength of evidence against the null. Decisions can be made by  thresholding the significance index. The two most widely used significance indices are the (two-sided) $p$-value and the local false discovery rate (\lfdr), which are respectively defined as 
\begin{eqnarray} \label{eq:lfdr}
P_i & = & 2 F_0(-|Z_i|), \nonumber \\
\lfdr(z) & = & \PP(\theta_i=0|Z_i=z) =(1-\pi)  f_0(z) / f(z).
\end{eqnarray}


\cite{SunCai07} showed that $p$-value based methods can be uniformly improved by the adaptive $z$-value (AZ) procedure that is based on ranking and thresholding the \lfdr. In this article, we aim to develop an \lfdr-based decision rule that is provably valid for FDX control. Such an approach has been taken by \cite{SunCai07} and more recently by \cite{basu2018weighted} and \cite{heller2021optimal}. Further, we want to ensure that the methodology is computationally efficient and can be seamlessly used in applications with millions of tests, analogous to  BH \citep{BenHoc95} and AZ \citep{SunCai07} for FDR control. 

\subsection{False Discovery eXceedance (FDX) and Power}

Consider a generic decision rule $\bm{\delta} =(\delta_1, \hdots, \delta_m)\in\{0, 1\}^m$. Using the notations in Section \ref{sec:notation}, the false discovery proportion can be defined as 
\begin{equation}\label{eq:fdp}
\mbox{FDP}\coloneqq \frac{\sum_i ( 1 - \theta_i) \delta_i }{(\sum_i \delta_i) \vee 1},
\end{equation}
where $a \vee b=\max(a,b)$. 
Let $\gamma$ represent a tolerance level on the FDP and $\alpha$ a small probability. An FDX procedure at level $(\gamma, \alpha)$ satisfies 
\begin{equation}\label{eq:fdx}
\mbox{FDX}\coloneqq \PP \left(\mbox{FDP}> \gamma\right)\leq\alpha. 
\end{equation} 
The efficiency or power of an FDX procedure is evaluated using the expected number of true positives:
\begin{equation}\label{eq:power}
\mbox{ETP}\coloneqq \EE \left(\textstyle\sum_i \theta_i \delta_i\right).
\end{equation}
We call an FDX procedure \emph{optimal} if it maximizes the ETP subject to the constraint \eqref{eq:fdx}. 

\section{Oracle Procedure for FDX Control} \label{sec:or}

This section considers an idealized setup where the distributional quantities $\pi$, $f_0$, and $f$ are known. Using these, we propose an oracle rule that controls the FDX at level ($\gamma, \alpha$) (Section \ref{sec:or-proc}) and then establish its validity and optimality (Section \ref{sec:or-prop}). The connection to the Poisson binomial distribution (PBD) is drawn in   Section \ref{sec:pbd}, and  computational shortcuts for fast implementation are developed in Section  \ref{sec:shortcuts}. The estimation of $\pi$, $f_0$, and $f$; and other implementation issues are discussed in Section \ref{sec:Implementation}.

\subsection{Oracle Procedure}\label{sec:or-proc}

Consider a decision rule that rejects $k$ hypotheses. Denote $\mathcal R_k$ the set of rejected hypotheses.  The number of false rejections is given by $\sum_{i\in \mathcal R_k} (1-\theta_i)$. Our derivation of the FDX procedure involves calculating the following tail probability 
\beq\label{tail-prob}
\PP_{\theta|\Z}(\mbox{FDP}>\gamma) = P_{\theta|\Z} \left\{\sum_{i\in \mathcal R_k} (1-\theta_i) > k\gamma\right\}.
\eeq
We shall see that \eqref{tail-prob} can be found using the PBD, which generalizes the binomial distribution to the case when each trial has a different probability of success. 

Denote the Poisson binomial distribution by $\PBD(k, \bm{p})$, with $k$ being the total number of trials and $\bm{p}= (p_i: i = 1, \hdots, k)$ the vector of success probabilities. To reflect that we consider the idealized setup, the notation $T_i^{OR}$ is used to denote the oracle \lfdr \ associated with study unit $i$:
\begin{equation}
\label{fdx:def:tor}
T_i^{OR} \coloneqq \PP(\theta_i = 0| Z_i = z_i) = \frac{(1-\pi)f_{0}(z_i)}{f(z_i)}.
\end{equation}

Procedure \ref{proc:proposedOR} describes our proposed oracle FDX procedure at level $(\gamma, \alpha)$.

\begin{proc}
	\label{proc:proposedOR} \phantom{bhoooot}
	\begin{enumerate}
		\item  Consider the lfdr statistic $(T^{OR}_i)_{i\in[m]}$ defined in \eqref{fdx:def:tor}, and denote $\left(T^{OR}_{(i)}\right)_{i\in[m]}$ the ranked statistic in ascending order. The corresponding null hypotheses are denoted $\left\{H_{(i)}^0: i\in [m]\right\}$.   
		\item Denote $\bm{p}^{(k)} = \left\{T^{OR}_{(1)}, \hdots, T^{OR}_{(k)}\right\}$. Let 
		$$K \coloneqq \max \left\{k: \PP\left(\PBD(k,  \bm{p}^{(k)})>\gamma k \right) \leq \alpha\right\},$$ and reject the top $K$ hypotheses $\{H_{(1)}^0, \cdots, H_{(K)}^0\}$ along the lfdr ranking.
		
	\item Denote $\mathcal R_k$ the set of rejected hypotheses. Reject $H_{(K+1)}^0$ with the following probability  
	\begin{equation}\label{rand:prob}
	\frac{\alpha - \PP_{\theta|\Z}\left \{\sum_{i\in\mathcal R_K} (1-\theta_i) > \gamma K \right\}}{\PP_{\theta|\Z} \left\{\sum_{i\in \mathcal R_{K+1}} (1-\theta_i) > \gamma(K+1)\right\} - \PP_{\theta|\Z} \left\{\sum_{i\in\mathcal R_K} (1-\theta_i) > \gamma K\right\}}. 
	\end{equation}
	\end{enumerate}
\end{proc}
\vspace{0.2in}


Similar to \cite{BenHoc95}, Procedure~\ref{proc:proposedOR} is a step-up procedure, in the sense that it starts from the least significant hypothesis (i.e., the one with the largest \emph{lfdr}) and moves up at each step to a more significant one. The procedure stops when it finds the first null hypothesis, ${H}_{(K)}^0$, for which the tail probability is less than $\alpha$. It then rejects all null hypotheses $\{{H}_{(1)}^0$, \ldots, ${H}_{(K)}^0\}$. The randomization of the decision at the last step ensures that we achieve \emph{exact} FDX control at the nominal level $(\gamma, \alpha)$. This randomization technique was employed in the weighted FDR procedure proposed by \cite{basu2018weighted} and later in \cite{gu2020invidious}. 

\subsection{Poisson Binomial Distribution (PBD) and its Connection to the \emph{lfdr}}\label{sec:pbd}

Poisson's binomial distribution or PBD refers to the sum of independent Bernoulli random variables, with not necessarily equal expectations. In the special case that the expectations are all equal, a PBD simplifies to a binomial distribution. The PBD provides a useful tool for probability calculations in a range of statistical applications (\citealp{chen1997statistical}). More recently, the PBD has been employed by \cite{dohler2020controlling} to develop FDX-controlling procedures for multiple testing with heterogeneous units. This section discusses the connection of the PBD to the \emph{lfdr} procedure.

Consider the two-group model \eqref{eq:mixmodel}. If we do not have any prior knowledge, then each null hypothesis is true (or false) with the same probability $1-\pi$ (or $\pi$).  However, conditional on the observables $\Z = (Z_i)_{i \in [m]}$, the unknown states of the hypotheses $\theta_i$ marginally follow Bernoulli distributions with heterogeneous success probabilities, namely, $\pi_i=\PP(\theta_i = 0| \Z)$. In the situation where the joint density $f(\Z|\theta_1, \cdots, \theta_m)$ can be factorized into the product of marginal densities. That is,
\[
f(\Z|\theta_1, \cdots, \theta_m)=\prod_{i\in[m]}f(Z_i|\theta_i), 
\]
and $\PP(\theta_i = 0| \Z)$ reduces to the \lfdr \ statistic $\PP(\theta_i = 0| Z_i)$. Furthermore, if $\theta_i$ are independent, then the partial sum $\sum_{i\in S}(1-\theta_i)$, where $S\subset [m]$ is an arbitrary nonempty subset of $[m]$, is a PBD random variable. 

In the oracle procedure (Procedure \ref{proc:proposedOR}), we rank hypotheses by the \emph{lfdr}; hence the subset $S$, which consists of the hypotheses we reject, is data-dependent. Then, the index set of the partial sum should be denoted by $S_{\Z}$, which critically depends on the thresholding (aka  ``selection") step of Procedure \ref{proc:proposedOR}. An important concern is that this may lead to a selection bias, but 
the next lemma proves that the selection based on \emph{lfdr} does not distort the Poisson binomial distribution of the partial sum $\sum (1-\theta_i) \mathcal{I}_{i \in S_{\Z}}$. Our proof essentially corroborates the general principle that selective inference based on Bayes rule is unbiased conditional on a selection event. The high-level idea is explained in earlier works  \citep[e.g.,][]{dawid1994selection}. 

\begin{lemma} (No Selection Bias) Ranking by the \emph{lfdr} does not alter the conditional distribution of the partial sums $\sum (1-\theta_i) \mathcal{I}_{i \in S_{\Z}}$, where $S_{\Z}$ denotes \textit{any} index set under consideration after viewing the data $\Z$.
\end{lemma}
\begin{proof}  Define random variables ${\blue R}_i := (1-\theta_i) \mathcal{I}_{i \in S_{\Z}}$, where $S_{\Z}\subset [m]$ is determined by the observables $\Z = (Z_i)_{i \in [m]}$.  Consider any nonempty fixed subset of indices denoted by $S_{\mathcal{P}}$. Consider two situations. 
\begin{itemize}
 \item (i) If $S_{\mathcal{P}} \nsubseteq S_{\Z}$, then there exists an index $i_0 \in S_{\mathcal{P}} \cap S_{\Z}^c$. We have
 $$
\EE\left\{\prod_{i \in S_{\mathcal{P}}} R_i | \Z\right\} = 0 = \EE \left(R_{i_0}|\Z\right) \cdot \EE\left(\prod_{i \in S_{\mathcal{P}} \backslash i_0} R_i | \Z\right).
$$ 
 
 \item (ii) If $S_{\mathcal{P}} \subseteq S_{\Z}$, then we have 
 $$
 \EE[\prod_{i \in S_{\mathcal{P}}} R_i | \Z]  = \EE[\prod_{i \in S_{\mathcal{P}}} {\blue (1-\theta_i)} | \Z] = \prod_{i \in S_{\mathcal{P}}} \EE [{\blue (1-\theta_i)} | \Z] = \prod_{i \in S_{\mathcal{P}}} \EE [R_i | \Z],
 $$ 
\end{itemize}

{\blue where the middle equality is due to factorization of the joint density and independence of the $\theta_i$s.} Therefore conditional on $\Z$, the $R_i$'s are identically zero if $i \notin S_{\Z}$ and independently distributed as Bernoulli random variables with the expectation of $lfdr(Z_i)$ if $i \in S_{\Z}$.
\end{proof}

\subsection{Properties of the Oracle Procedure}\label{sec:or-prop}

The oracle procedure can be split into two steps, ranking and thresholding, which shape the procedure's properties. For the ranking step, which orders the hypotheses from the most significant to the least significant according to a significant index, a desirable property is optimal ranking.  The optimal ranking with respect to a significant index $T$ indicates that the thresholding rule along the $T$ ranking has equal or larger power than any other rule at the same error rate. 

For the thresholding step, which  chooses a cutoff along the ranking to control the desired error rate, a desirable property is error-rate control at the target level. In this section, {\blue assuming that lfdr ranking does not introduce any selection bias (see Lemma 1)}, we first show that the FDX level $(\gamma, \alpha)$ is exhausted by our proposed oracle procedure (thereby achieving the second property), and then establish that the \lfdr \ ranking is optimal (thereby achieving the first property).   
 
\begin{proposition} (Exact Control of FDX).
	Procedure~\ref{proc:proposedOR} controls the FDX at level $(\gamma, \alpha)$.
\end{proposition}

\begin{proof} By design, Procedure 1 ensures $\PP_{\theta|\Z} \left\{\sum_{i\in\mathcal R_K} (1-\theta_i) > \gamma K\right\} \leq \alpha$. Consider an independent arbiter $U$ (i.e. a weighted coin flip), which is chosen to favor one more rejection with probability given in \eqref{rand:prob}. 
This arbitrary randomization at Step 3 in Procedure~\ref{proc:proposedOR} ensures that 
\[
E_{(\theta, U)|\Z} \left[\mathbb{I}_{\sum_i (1-\theta_i) \delta_i^{*} > \gamma \sum_i \delta_i^{*}} \right] = \alpha, 
\]
where $\delta_i^{*}$ is the decision rule determined by the data $\Z$ and the independent arbiter. Taking a further expectation with respect to $\Z$ completes the proof. 
\end{proof}

\begin{proposition} (Optimal Ranking)
In the i.i.d.~two-group model, Procedure~\ref{proc:proposedOR} has the best ranking almost surely in the sense that for any other decision rule at FDX-level $(\gamma, \alpha)$, we can always find an lfdr-based thresholding rule at the same level that has a higher or equal ETP. 
\label{Prop:OR}
\end{proposition}

The proof of Proposition 2 is presented in the supplementary material.

\subsection{Computational Shortcuts}\label{sec:computational-shortcuts}\label{sec:shortcuts}

Procedure~\ref{proc:proposedOR} is a step-up procedure, which starts by computing the tail probability of the PBD for all tests under consideration. At each progressive step, the set of tests under consideration decreases. However, if we have a massive number of tests to begin with, then Procedure~\ref{proc:proposedOR} can be computationally intensive. To facilitate fast implementation in large-scale testing problems, we develop some computational shortcuts and modify Procedure \ref{proc:proposedOR} as follows. 

\begin{proc} \label{proc:proposedfaster} \phantom{}
Consider the lfdr statistics $(T^{OR}_i)_{i\in[m]}$ defined in \eqref{fdx:def:tor}.  The modified FDX procedure consists of 4 steps that successively narrow down the focus of the search. 
	\begin{enumerate}
	\item Order the lfdr statistics in ascending order. Denote the ordered statistics by $(T^{OR}_{(i)})_{i\in[m]}$. 
		\item Reject up to $K_1 := \max \left[k \in [m]:~\EE_{\theta|\Z} \left\{\sum_{i=1}^k (1-\theta_i) \right\} \leq k \cdot \{\alpha + \gamma (1 - \alpha)\}\right]$.
		\item Reject up to $K_2 := \max \{k \in [K_1]: \PP \left(Y > \gamma k \right) \leq \alpha\}$ where,\\
		$Y \sim \mbox{Binomial}\left(k, (\prod_1^k T_{(i)})^{1/k}\right)$.
		\item Denote $\bm{p^{(k)}} = (T^{OR}_{(1)}, \hdots, T^{OR}_{(k)})$. Reject only up to 
		$$K := \max \left\{k \in [K_2]: \PP\left(\PBD(k,  \bm{p^{(k)}})>\gamma k \right) \leq \alpha \right\}. $$ 		
	\end{enumerate}
\end{proc}

First, note that relative to Procedure~\ref{proc:proposedOR}, Procedure~\ref{proc:proposedfaster} has two additional steps (i.e., Step 2 and Step 3). We progressively reduce the number of hypotheses under consideration so that the computationally intensive Step 4 can be conducted on a much smaller subset of hypotheses.  This greatly reduces the computational cost because Step 2 only involves computing cumulative average \emph{lfdr} over progressively smaller sets of hypotheses, and Step 3 involves binomial calculations rather than Poisson binomial calculations over sets of hypotheses which are further whittled down.

Next we justify the computational shortcuts by showing that Procedures~\ref{proc:proposedOR} and \ref{proc:proposedfaster} are equivalent. The main idea is that both Steps 2 and 3 are step-up procedures in themselves:~a testing unit's failure to meet their criteria guarantees that the testing unit will also fail the tail probability criteria in Step 4. 

Step 2 is equivalent to the adaptive $z$-value procedure for FDR control in \cite{SunCai07} at the level $\alpha+\gamma(1-\alpha)$. We claim that $\PP_{\theta|\Z} \left\{\sum_{i=1}^k (1-\theta_i) > \gamma k \right\} \leq \alpha$
implies
\[
\EE_{\theta|\Z} \left\{\sum_{i=1}^k (1-\theta_i)\right\} \leq k \cdot \{\alpha + \gamma (1 - \alpha)\}. 
\]
To see this, rewrite the above expectation by partitioning the event into the sub-event where the number of false positives is greater than $\gamma k$ and its complement event:
$$
\EE_{\theta|\Z} \left(\mathcal{I}_{\sum_{i} (1-\theta_i) > \gamma k} \sum_{i = 1}^k (1-\theta_i)\right)  +  \EE_{\theta|\Z} \left(\mathcal{I}_{\sum_{i} (1-\theta_i) \leq \gamma k} \sum_{i=1}^k (1-\theta_i)  \right).
$$
The sum in the first expectation is bounded by $k$, and the sum in the second expectation, the number of false positives, is bounded by $\gamma k$. {\blue Assume that there is an $\alpha'$, such that $P_{\theta|Z_1 \dots Z_n} \{ \sum_{i = 1}^k (1-\theta_i) > \gamma k \} = \alpha'$, then the sum of the two expectations is bounded by $k \cdot \{\alpha' + \gamma (1 - \alpha')\} = k \gamma + \alpha' \cdot (k - \gamma k)$. Since $ (k - \gamma k) \geq 0$  and $\alpha' \leq \alpha$, because   $P_{\theta|Z_1 \dots Z_n} \{ \sum_{i = 1}^k (1-\theta_i) > \gamma k \} \leq \alpha$, then $k \cdot \{\alpha' + \gamma (1 - \alpha')\} \leq k \cdot \{\alpha + \gamma (1 - \alpha)\}$, which becomes the upper bound.

  }Thus if the condition in Step 4 holds, then the condition in Step 2 also holds. It follows that if the condition in Step 2 fails, then the condition in Step 4 also fails. Hence the reduction of the set of hypotheses produced by Step 2 is legitimate, as it only eliminates cases in which the condition of Step 4 would fail. Step 2 ends when we find the largest index for which the condition does not fail, and we pass $\{\mathcal{H}_{1}$, \dots, $\mathcal{H}_{K_1}\}$ to the next step in the procedure.

In Step 3, we apply a useful result from \cite{shaked2007stochastic}. Concretely, consider $n$ independent Binomial random variables $X_i \sim B(1, p_i)$ with $i \in\{1, \dots, n\}$, and 
$$Y \sim \mbox{Binomial}\left(n, (\Pi_{i=1}^n p_i)^{1/n}\right).$$ 
Then $Y$ is stochastically smaller than $\sum_{i\in[n]} X_i$. This implies that if the condition of Step 3 fails, i.e.,  $$\PP \left\{\mbox{Binomial}\left(k, (\prod_i^k T^{OR}_{(i)})^{1/k}\right) > \gamma k \right\} > \alpha,$$ then the condition in Step 4 also fails (i.e., {\blue $\PP\left(\PBD(k, \bm{p^{(k)}}) >\gamma k \right) > \alpha$}), where $\bm{p^{(k)}} = (T^{OR}_{(1)}, \hdots, T^{OR}_{(k)})$. Step 3 is also a step-up search, which will end at the first index for which the condition does not fail, denoted by $K_2$. 

Finally, Step 4 of Procedure~2 is equivalent to Step 2 of Procedure~1, but it is applied to a set with $K_2$ hypotheses instead of the initial set of $m$ hypotheses. The equivalence between  Procedures~\ref{proc:proposedOR} and \ref{proc:proposedfaster} is thus established. The computational advantage of Procedure \ref{proc:proposedfaster} is illustrated in the supplementary material.


\section{Implementation and Related Issues}\label{sec:Implementation}

\subsection{Estimation of the \emph{lfdr}}

So far, we have assumed that the \emph{lfdr} is known. However, it is unknown in practice and must be estimated from data. There are several approaches to estimating the \emph{lfdr}, which differ in how the null density $f_0$, the null proportion $(1-\pi)$, and the mixture density $f$ are computed. We use the model-based clustering approach of \cite{fraley2002model} to estimate the \emph{lfdr} in our simulation experiments. The approach is available in the R package \textit{mclust} (version of November 20, 2020). The package first clusters the data into $G$ groups then computes for each $z$-value the probability that it belongs to a specific group $(\pi_g, g\in [G])$, and finally estimate the corresponding cluster-wise densities $(f_g: g\in[G])$. The group with the highest $\pi_g$ will be defined as the null group, and we denote the corresponding proportion and density as $\pi_0$ and $f_0$. In our numerical studies, $f_0$ always corresponds to the cluster whose center is close to zero. The \emph{lfdr} is then computed as $$lfdr = \frac{\pi_0 f_0}{\sum_{g \in [G]} \pi_g f_g}.$$ Alternatively, the denominator can also be computed as the mixture density $f(\cdot)$ using a non-parametric kernel density estimator.

Another method for computing the unknown \emph{lfdr} statistic, particularly in applications where the data generating process is blurry or unknown, is through the R package \textit{locfdr} based on \cite{efron2004large,efron2008microarrays,efron2009empirical}. We use this approach in the real data application presented in Section~\ref{Sec.Application}. Other methods include the adaptive $z$-value approach discussed in \cite{SunCai07}, for which code is available at \url{http://stat.wharton.upenn.edu/~tcai/paper/html/FDR.html}, and the robust error-correction method recently proposed by \cite{roquain2020false}. 



\subsection{Dependencies and Exchangeability}\label{Sec.Exchangeability}

The proof of the optimality of the lfdr ranking (Proposition~\ref{Prop:OR}) requires that the observables $(\Z=Z_{i\in[m]})$ be  independent. This section investigates the optimality issue under dependence. We first show that ranking by \emph{lfdr} is still optimal when hypotheses are jointly Gaussian and exchangeable  (i.e., equal-variances and equal-covariances):~$\Z|\pmb \theta \sim \mathcal{N} (\mu \ \pmb \theta, \Sigma$), where $\mu$ is the common mean, and the correlation matrix corresponding to $\Sigma$ is equi-correlated: 
{\blue
\[
\Sigma  = 
\begin{bmatrix}
1 & \rho & \ldots&\rho\\
\rho & 1 & \ddots&\vdots\\
\vdots & \ddots & \ddots &\rho\\
\rho&\ldots&\rho & 1
\end{bmatrix}
\]
where $\rho \in [0,1)$. This model can rewritten as (\cite{carpentier2021estimating})
\[
Z_i = \mu \theta_i + \rho^{1/2}W + (1-\rho)^{1/2}\zeta_i,
\]
for $1 \leq i \leq m$ and $W, \zeta_i$ are independent $N(0, 1)$. Then conditional on W the individual $Z_i$ are independent and the null mean and variance are $\rho^{1/2}W$ and $1-\rho > 0$.} In supplementary Section \ref{Sec.Counterexample}, by contrast, if hypotheses are non-exchangeable, we provide a counter-example in which the optimality of the \lfdr \ ranking does not hold (the numerical results are provided in Table \ref{Tbl.Counterexample}).

\begin{proposition} (Optimal Ranking for Exchangeable Gaussian Observations) When $\Z|\pmb \theta \sim \mathcal{N} (\mu \ \pmb \theta, \Sigma$), Procedure~\ref{proc:proposedOR} has the best ranking for almost all sample points:~For any decision rule with FDX-level $(\gamma, \alpha)$ we can find an lfdr-based thresholding rule at the same level that has a higher or equal ETP.
\end{proposition}

The proof of Proposition 3 is presented in the supplementary material.

After ranking the hypotheses according to the \emph{lfdr}, we must choose a threshold to control the FDX. Procedure 2 requires that hypotheses be independent. Hence if we apply the method to a data set where the test statistics are dependent, our FDX procedure may not provide the optimal threshold. If the source or nature of the dependency {\blue structure is known}, one could still compute a proper posterior probability by evaluating the probabilities of all possible combinations of $\pmb \theta$ conditional on data. However, this evaluation can be computationally prohibitive. The computational burden can be overcome if we assume that individual $z$-values are independent conditionally on the unknown location and scale of the data-generating process. Our strategy is to first estimate the hyper-parameters, thus obtaining an estimated ``null'' distribution, and then continue as if the tests were independent. We return to this point in Section~\ref{Sec.Dependencies}. 

\section{Numerical Experiments}\label{sec:numeric}

We present several numerical examples to highlight the properties of Procedure 2 and illustrate how it compares to other competing methods. We start with a setup based on the two-group model \eqref{eq:mixmodel} that assumes independence, then turn to a setup that incorporates dependency. In the supplementary material, we simulate a setup that mirrors the real-life application presented in Section~\ref{Sec.Application}. {\blue In the implementation of our procedures we compute PBDs following the method described in \cite{Biscarri2018} as implemented in the R Package ‘PoissonBinomial’.}

\subsection{Independent Hypotheses}\label{Sec.IIDsim}

We consider a setup where the observed data are generated from the two-group model \eqref{eq:mixmodel}. The non-null proportion is given by $\pi \in \{0.1, 0.2, 0.3\}$, the null distribution is $\mathcal N(0,1)$, and the non-null distribution is $\mathcal{N}(\mu, 1)$, with $\mu \in \{-1.5, -2, -2.5\}$. This is a standard setting; for example, the case with $\pi = 0.2, \ \mu = -2$, corresponds to the example in Table 1 of \cite{heller2021optimal}. 

We simulate $10,\!000$ data sets, each with $m=5,\!000$ tests. We evaluate Procedure 2 for ($\gamma = 0.05$, $\alpha = 0.05$), and compare it to a set of representative procedures:~\cite{SunCai07} (SC), \cite{BenHoc95} (BH), \cite{GuoRomano07} (GR), and \cite{LehmannRomano05} (LR). {\blue All SC procedures reported here, as a well-justified contender to the BH method, use the estimated lfdr statistic.} For Procedure 2 we consider three versions: the oracle, which knows the non-null distribution and the non-null proportion (Oracle), the version where the distributional parameters are unknown and need to be estimated ($lfdr$), and the version  where the distributional parameters are unknown but a strong prior of 1 is imposed on the proportion of nulls ($lfdr(\hat{\pi} = 0))$, which can be viewed as a conservative version of the \lfdr \ procedure. The FDX and ETP levels of different methods are obtained by averaging the testing results from the $10,\!000$ data sets. We summarize the empirical performances of different methods in Table~ \ref{Tbl.Comparison}. In the FDX row, the numbers represent  the relatively frequencies among the $10,\!000$ experiments where the realized FDP exceeds $\gamma$.

The following three general patterns can be observed across various simulation scenarios:
\begin{enumerate}

\item The BH and SC procedures, which are designed to control FDR, are effective for FDR control, with SC close to the desired level $\gamma$, and BH close to $(1-\pi)\gamma$. Both procedures have highly inflated FDX especially for weak and sparse scenarios:~A key motivation for our work.

\item The procedures that are designed to control FDX have varying degrees of success. The FDX level of Procedure 2 yields an $\alpha$ that is generally close to the 5\% target. By contrast, GR and LR become {\blue more conservative} as sparsity decreases (i.e., $\pi$ increases) and as the absolute value of the non-null effect becomes larger.

\item The FDX methods are in general less powerful than the  FDR procedures. Procedure 2 is more powerful than GR and LR in most settings, and is relatively closer to BH in power than the other two FDX methods are in some settings. 

\end{enumerate}

\begin{boxedtable}{Comparison of different procedures}{Tbl.Comparison}
\tablexplain{The table compares the performance of Procedure 2 relative to some popular methods:~\cite{SunCai07} (SC), \cite{BenHoc95} (BH), \cite{GuoRomano07} (GR), and \cite{LehmannRomano05} (LR). Three versions of Procedure 2 are implemented:~the oracle version (Oracle), a version where the parameters of the data generating process are estimated from the data ($lfdr$), and a version where we impose the assumption that $\pi = 0$, $lfdr(\hat{\pi} = 0)$. The data-generating process is a mixture model, where with probability $1-\pi$ the test is drawn from a $\mathcal{N}(0,1)$ (null), and with probability $\pi$ the test is drawn from $\mathcal{N}(\mu, 1)$ (alternative), where $\pi \in \{0.1, 0.2, 0.3\}$ and $\mu \in \{-1.5, -2, -2.5\}$. Each simulation considers 5000 tests. We repeat the exercise for 10000 simulations. FDX control is implemented at $\gamma = 0.05$ and with a confidence level $1-\alpha = 0.95$. FDR control is implemented at a nominal level of $\alpha = 0.05$. } 
\begin{ctabular}{l l rrrrrrr} 
\midrule
&&&&&&\multicolumn{3}{c}{Procedure 2}\\
\cmidrule{7-9}
& & SC & BH &  GR & LR & Oracle & $lfdr$ & $lfdr(\hat{\pi} = 0)$\\  
\midrule

$\pi = 0.2$ & FDX & 0.452 & 0.348 & 0.040 & 0.040 & 0.047 & 0.076 & 0.062\\ 
$\mu = -1.5$& FDR & 0.050 & 0.040 & 0.013 & 0.013 & 0.015 & 0.021 & 0.019\\ 
& Power (\%) & 3.4 & 2.6 & 0.3 & 0.3 & 0.3 & 0.4 & 0.4\\ 
\midrule 
$\pi = 0.2$ & FDX & 0.484 & 0.242 & 0.039 & 0.037 & 0.052 & 0.066 & 0.025\\ 
$\mu = -2$& FDR & 0.050 & 0.040 & 0.004 & 0.004 & 0.028 & 0.029 & 0.022\\ 
& Power (\%) & 22.6 & 18.8 & 1.5 & 1.2 & 13.5 & 14.0 & 10.6\\  
\midrule 
 $\pi = 0.2$& FDX & 0.471 & 0.133 & 0.014 & 0.000 & 0.049 & 0.054 & 0.006\\ 
$\mu = -2.5$& FDR & 0.050 & 0.040 & 0.030 & 0.002 & 0.036 & 0.036 & 0.028\\  
& Power (\%) & 51.0 & 46.5 & 40.6 & 7.6 & 44.3 & 44.3 & 39.6\\ 
\midrule
$\pi = 0.1$ & FDX & 0.458 & 0.399 & 0.044 & 0.042 & 0.047 & 0.068 & 0.059\\ 
$\mu = -2$& FDR & 0.049 & 0.045 & 0.007 & 0.007 & 0.011 & 0.016 & 0.014\\ 
& Power (\%) & 11.2 & 10.3 & 1.2 & 1.2 & 2.3 & 3.1 & 2.5\\ 
\midrule 
$\pi = 0.3$ & FDX & 0.484 & 0.061 & 0.019 & 0.019 & 0.053 & 0.059 & 0.004\\ 
$\mu = -2$& FDR & 0.050 & 0.035 & 0.011 & 0.002 & 0.035 & 0.035 & 0.023\\  
& Power (in \%) & 33.2 & 25.8 & 8.8 & 1.4 & 25.7 & 25.9 & 18.7\\ 
\midrule 
\end{ctabular}
\end{boxedtable}

Finally, we compare the performance of the three versions of Procedure 2. When the $lfdr$ is estimated, the procedure generally delivers higher levels of FDX, particularly when signals are sparse and weak. To address the FDX inflation in  situations where the signals are sparse and weak, we propose  a conservative approach which imposes the assumption that the null proportion is approximately 1 [$lfdr(\hat{\pi} = 0)$]. We see that this approach leads to conservative FDX control but is still more powerful than GR and LR. 

\subsection{Impact of Dependence}\label{Sec.Dependencies}

As mentioned in Section~\ref{Sec.Exchangeability}, Procedure 2 is not guaranteed to control the FDX when the hypotheses are non-exchangeable. However, in particular cases which are useful in practice, Procedure 2 can still achieve excellent performance by using its empirical Bayes framework to estimate the $lfdr$ statistic in order to learn the structure of the null distribution. 

First, we illustrate the impact of dependence on the distribution of Student t-statistics using scenarios which may reasonably reflect dependence structures in real life.  Figure~\ref{Fig.Scatter1} shows a scatter plot of 100 pairs of two null test statistics under three settings. In the left panel, the two null test statistics are independently generated. In the middle panel, the data are dependent, because they are  generated as the sum of an independent $\mathcal N(0, 1)$ variable and a noise variable following an equi-correlated structure with $\rho = 0.25$. In the right panel, the observed data are again dependent but follow a hierarchical setup. In this setup $\pi=0.1$, and non-null $z$-values can be divided into two groups: group 1 follows the marginal distribution $\mathcal{N}(0.25, 1)$ with probability $\pi^+=0.05$ and group 2 follows the marginal distribution $\mathcal{N}(-0.25, 1)$ with probability $\pi^-=0.05$. Under the null, observations arise from a perturbed variation of the standard normal $\mathcal{N}(\mu, 1)$, where $\mu \sim Unif[-0.1, 0.1]$. Note that conditional on the realized value of $\mu$, the observations are independent; unconditionally, the null observations are not. Specifically, the correlation among the Student test-statistics is about 0.25.  Figure~\ref{Fig.Scatter1}  shows a distinctly different distribution for the independent statistics from the distribution for dependent statistics. However, aside from the third scatter plot being more spread out due to the location shift, the patterns in the last two plots look strikingly similar. 

\begin{boxedfigure}{Conditional dependencies}{Fig.Scatter1}
\tablexplain{The figure visually compares different distributions of test statistics:~the leftmost panel corresponds to mutually independent tests; the middle panel shows tests statistics obtained from data that have a 0.25 correlation caused by common additive noise; the rightmost panel, instead, shows test statistics obtained from conditionally independent samples, which are however unconditionally correlated because of a common location shift that induces a correlation of about 0.25.}
\begin{center}
\includegraphics[width=0.8\textwidth]{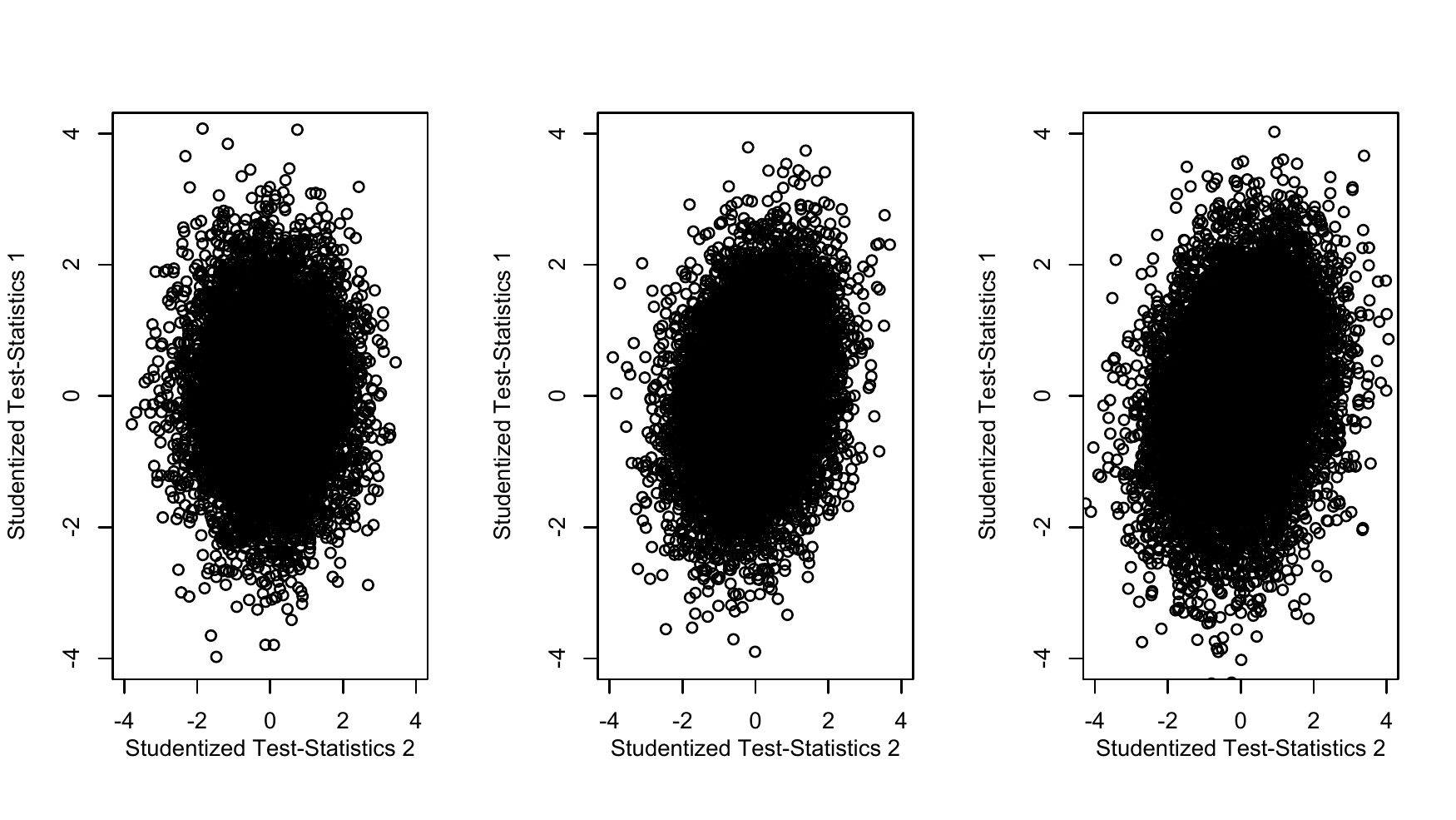}
\end{center}
\end{boxedfigure}

In our numerical experiments, we simulate data for $m=5,\!000$ tests. The goal is to control the FDX at $\gamma = 0.1$ with confidence level $1-\alpha = 0.95$. We compare our proposed procedure to the method in \cite{GuoRomano07}  (GR), and the bootstrap-based procedure in \cite{romano2007control} and \cite{romano2008formalized} (RSW). Notably, RSW is shown to asymptotically control the FDX in the presence of arbitrary dependencies.

Table~\ref{Tbl.Scatter1} reports the realized FDX, average power (correctly rejected hypotheses as a percentage of the number of non-nulls), and the average of the realized FDPs over $10,\!000$ experiments. In Columns 3-4, we tabulate results obtained for GR and RSW, both of which assume the theoretical null (i.e., $\mathcal{N}(0,1)$). Columns 5-6 provide corresponding statistics when the null distribution is estimated from the data (i.e., the empirical null). We correct for the unknown null location by re-centering all the test statistics using the mean of all observations. Note that Procedure 2 is designed to adapt to the scenario automatically so that there is no essential difference when re-centering the data. The $lfdr$ statistic is estimated using the $locfdr$ package \citep{efron2004large}. 

We can see that both GR and RSW yield FDX levels that are much higher than the nominal level. Both procedures do better in FDX control when the estimated empirical null is used (columns 5-6), although RSW becomes very conservative. By contrast, Procedure 2 is capable of controlling the FDX at the desired level while maintaining good power.

\begin{boxedtable}{Conditional dependencies}{Tbl.Scatter1}
\tablexplain{The table reports the average FDX, FDR, and power for 10000 simulated experiments. Each experiment considers {\blue 4000} tests drawn in a hierarchical setup where a hypothesis is null most of the time, with a probability of 0.90. When non-null, one-half of the time, the marginal distribution is $\mathcal{N}(0.25, 1)$, and one-half of the time it is $\mathcal{N}(-0.25, 1)$. When null, the observations arise from a perturbed variation of the standard normal,  $\mathcal{N}(\mu, 1)$ where $\mu \sim Unif[-0.1, 0.1]$. We compare Procedure 2 {\blue (Data-Driven)} to \cite{romano2008formalized}  (RSW) and \cite{GuoRomano07}  (GR). To facilitate comparison, the table shows the results from applying the original RSW and GR procedures (Theoretical Null) alongside the case where the econometrician can estimate the null distribution from the data (Empirical Null). FDX control is implemented at $\gamma = 0.1$ and with a confidence level $1-\alpha = 0.95$.}
\begin{ctabular}{l l rrrrrrr} 
\midrule
&&\multicolumn{2}{c}{Theoretical Null}&\multicolumn{2}{c}{Empirical Null}\\
\cmidrule{3-4} \cmidrule{5-6}
& Procedure 2 & RSW & GR &  RSW & GR\\  
\midrule
FDX & 0.046 & 0.492 & 0.771 & 0.003 & 0.164 &\\ 
FDR & 0.057 & 0.131 & 0.223 & 0.016 & 0.079\\ 
Power (\%) & 28.5 & 18.2 & 32.6 & 10.7 & 32.2\\ 

\midrule
\end{ctabular}
\end{boxedtable}

\section{Application:~Financial Trading Strategies}\label{Sec.Application}

We apply Procedure~\ref{proc:proposedfaster} to a real-life example where the goal is to identify interesting trading strategies among over two million candidate strategies, as in \cite{chordia2020anomalies}. The construction of trading strategies reflects exactly the simulation setup of Section~\ref{Sec.FinSimulation} in the supplementary materials. Each trading strategy is constructed by sorting stocks into deciles based on a trading signal at the end of June of each year. Stocks in the top decile are purchased at the closing price, and stocks in the bottom decile are sold short. Portfolio compositions are held for twelve months, but the weights are rebalanced monthly to reflect value-weighted exposures (i.e., stocks weights are proportional to relative market capitalizations). As trading signals, we consider every variable in the combined COMPUSTAT/CRSP datasets:~we take the level, the growth rate, the ratio between two variables, and a transformation of three variables (i.e., (x1 - x2)/x3). When more than one variable is involved, we consider all possible combinations. We apply filters to guarantee that strategies are well populated and that microstock returns are not overly biased. Details can be found in \cite{chordia2020anomalies}. In total, we obtain 2,396,456 trading strategies for the period between 1972 and 2015.

The data generating process provides that asset (stocks or portfolios) returns arise because of three components:~a systematic risk-premium, an idiosyncratic mean-zero, and a time-invariant component ($\alpha_i$). Under the null, the returns are entirely due to compensation for exposure to systematic risk factors; the time-invariant component is precisely zero (i.e., $\alpha_i =0$). We test whether the portfolio $\alpha_i$ is zero (i.e., this is a two-tail test) for 2,396,456 strategies. Thus we have a standard multiple testing problem.

The rationale behind adopting a combinatorial approach to constructing trading signals is essentially twofold:~it accounts for both variables used in common practice and variables left out by common practice. On the one hand, there is a long tradition among finance scholars and practitioners to relate stock returns to accounting variables:~quantities such as the equity market value of a firm (i.e., a level), the profitability of the assets (i.e., a ratio of two), and the ratio of assets minus equity, divided by assets (i.e., a transformation of three) have all been studied as predictors of future stock returns. See, for example, \cite{ChenZimmerman2021} who construct a large laboratory dataset of such variables. On the other hand, only predictors that worked and were discovered by academics or those that no longer worked and were found by industry practitioners are known. That leaves a large set of possible trading signals:~those that were tried by academics or practitioners but did not work, those still used in the industry but that are not widely publicized (for obvious reasons), and those that were never tried in the first place. In other words, there is a significant file drawer problem by considering many trading signals of the same functional form as those that have likely been studied. The combinatorial approach aids in providing an exhaustive set that can be analyzed through the lens of a multiple testing procedure.  Adopting the combinatorial approach to generating trading signals is not without consequences:~we will uncover many trading strategies that appear to be very profitable, but which likely are also meaningless, in the sense that they are artifacts of our data mining, are not based on any reasonable economic argument, and are likely not going to be profitable out of sample.

As mentioned above, this set of trading strategies has been studied in \cite{chordia2020anomalies}. The authors apply several multiple hypotheses procedures and still ``find'' many profitable strategies before applying some economic restriction. We expect the proportion of signals that are true predictors to be tiny. Thus, we would expect that a proper multiple testing procedure fails to select the very great majority of the strategies. Probably a problematic aspect of that study is the failure to incorporate the information gained from the data about the null distribution into the procedures. Thus, this particular data set seems perfect to evaluate Procedure 2, which heavily relies on ranking hypotheses by the data-driven \lfdr \ (i.e., the one that relies on the empirical null). 

\begin{boxedfigure}{{\blue Data representation under different assumptions}}{Fig.histograms}

\tablexplain{The figure compares the histogram of the distribution of $t$-statistics for 2,396,456 trading strategies to the $\mathcal{N}(0,1)$ and the density of the estimated empirical null. We use the analytical method of estimating the empirical null distribution parameters and null proportion described in Section 4 of \cite{efron2008microarrays}. The data contains 2,396,456 trading strategies for the period between 1972 and 2015.}

\vspace{-1.5in}
\begin{center}
\includegraphics[width=0.8\textwidth]{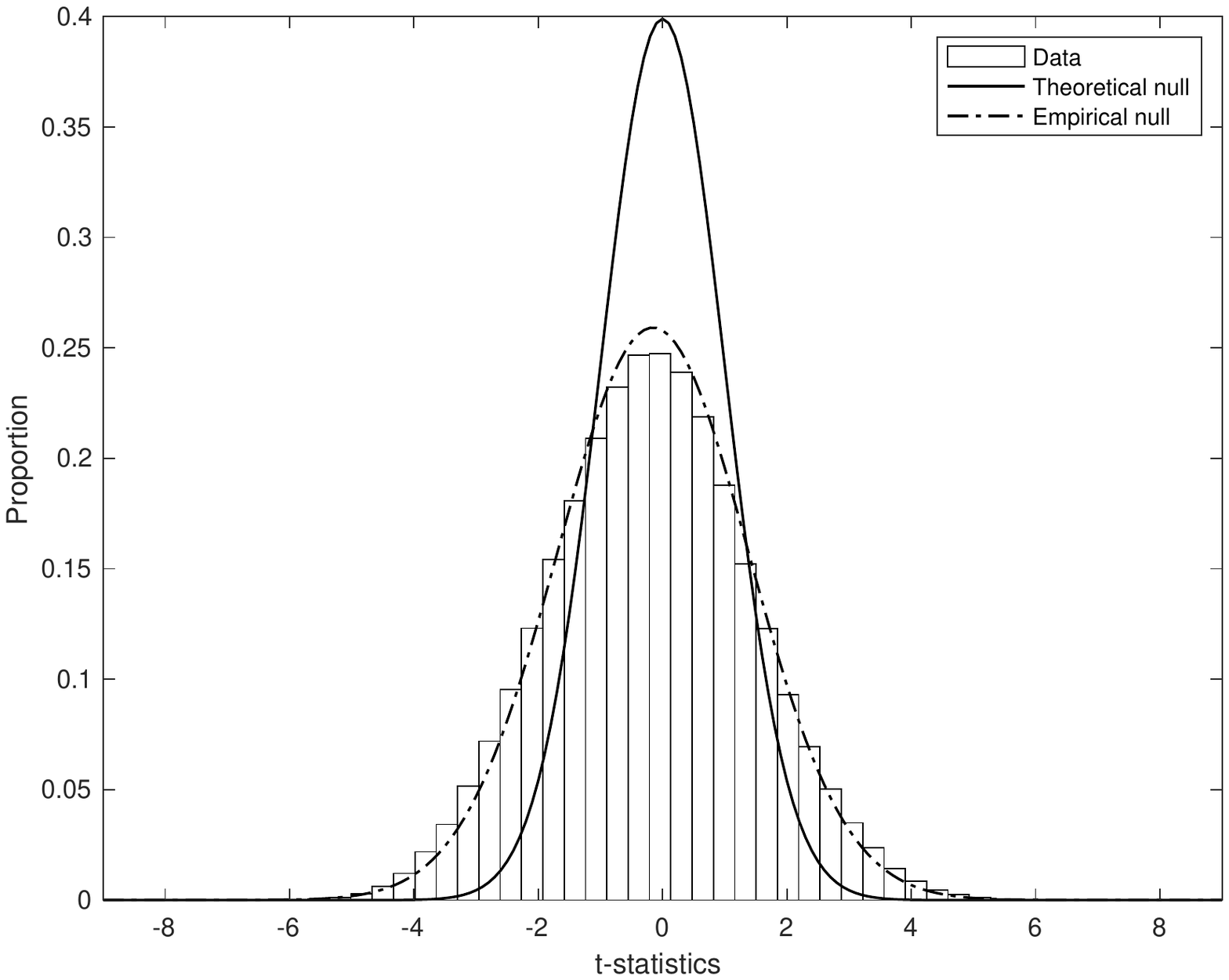}
\end{center}

\vspace{-1.5in}
\end{boxedfigure}

In Figure~\ref{Fig.histograms} we compare the histogram of the distribution of 2,396,456 alpha $t$-statistics, with the theoretical null (i.e., a normal with mean zero and standard deviation equal to 1) and with the empirical null. We estimate the empirical null distribution parameters and null proportion using the analytical method described in Section 4 of \cite{efron2008microarrays}.

The data is more spread out than the theoretical null but relatively close to the empirical null. The cross-sectional distribution of estimated alphas is dependent because some signals are correlated, and alpha is conditional on a set of common returns. In that sense, conceptually, the Efron empirical null is a much better approximation to the data-generating model. However, because our data contains many trading strategies that have already been reported as profitable, we expect some amount of divergence in the tails of the respective distributions (i.e., few genuinely non-zero alphas). How many non-nulls may be in the data is a question that can only be answered by correcting for multiple hypotheses.

We implement Procedure 2 and report in Table~\ref{Tbl.TradingStrategies} the number of strategies that are selected for different levels of $\gamma$ and $\alpha$, where $\gamma$ denotes the maximum allowable proportion of false discoveries (FDP), and $\alpha$ refers to the allowable tail probability. Similar to what we do for the simulation exercise described in Section~\ref{Sec.FinSimulation} in the supplementary materials, we compare the results obtained from applying Procedure 2 to those obtained  by applying the \cite{GuoRomano07} procedure (GR), the FDP-StepM procedure of \cite{romano2007control} and \cite{romano2008formalized} (RSW), and the FDR procedure of \cite{SunCai07} (SC).

\begin{boxedtable}{Discoveries in a sample of 2,396,456 trading strategies}{Tbl.TradingStrategies}
\tablexplain{The table reports the number of trading strategies selected by Procedure 2, \cite{GuoRomano07} (GR), the FDP-StepM procedure of \cite{romano2007control} and \cite{romano2008formalized} (RSW), and the FDR procedure of \cite{SunCai07} (SC) when applied to the set of 2,396,456 stock trading strategies constructed from accounting and stock price information for the period between 1972 and 2015.}

\begin{ctabular}{l r r r  c  r r r}
\midrule
&\multicolumn{7}{c}{Panel A: Procedures based on empirical null}\\
\midrule
&\multicolumn{3}{c}{Procedure 2}&&\multicolumn{3}{c}{GR}\\
$\gamma / \alpha$&0.01&0.05&0.10&&0.01&0.05&0.10\\
\cmidrule{2-4} \cmidrule{6-8}
0.05&1&4&5&&1&2&3\\
0.10&1&20&24&&1&2&3\\
0.20&40&47&51&&1&2&3\\
\midrule
&\multicolumn{7}{c}{Panel B: Procedures based on theoretical null}\\
\midrule
&\multicolumn{3}{c}{Procedure 2}&&\multicolumn{3}{c}{GR}\\
$\gamma / \alpha$&0.01&0.05&0.10&&0.01&0.05&0.10\\
\cmidrule{2-4} \cmidrule{6-8}
0.05&253,837&254,894&255,454&&204,011&205,113&205,654\\
0.10&416,217&417,137&417,627&&351,377&352,334&352,907\\
0.20&698,052&698,881&699,322&&626,247&627,418&628,007\\
\midrule
&\multicolumn{7}{c}{Panel C: Alternative procedures}\\
\midrule
&\multicolumn{3}{c}{RSW}&&\multicolumn{3}{c}{SC}\\
$\gamma / \alpha$&0.01&0.05&0.10&&0.01&0.05&0.10\\
\cmidrule{2-4} \cmidrule{6-8} 
0.05&5,528&32,812&65,001&&290,932&411,855&549,435\\
0.10&21,867&90,722&15,1614&&446,077&549,435&673,698\\
0.20&96,241&235,708&328,007&&722,506&808,424&915,176\\
\midrule
\end{ctabular}
\end{boxedtable}

In general, the number of findings increases with how many false discoveries are allowed. For example, for a choice of $\gamma = 0.10$ and $\alpha = 0.05$, we pick out 20 strategies, while for a $\gamma = 0.2$ the procedure selects 47 strategies. The number of selected strategies also varies considerably with the assumption about the shape of the null hypothesis:~if one relies on the theoretical null (Panel B) instead of Efron's empirical null, the number of discoveries grows dramatically.

Procedure 2 selects more strategies than the standard frequentist procedure of  \cite{GuoRomano07}, which emerges from our simulation as the most powerful solution in the frequentist's paradigm. This is not surprising as the result of the simulation presented in the previous section confirms GR to be less powerful. For example, in the case where the empirical null is used, $\gamma = 0.10$, and $\alpha = 0.10$, Procedure 2 selects 24 strategies while GR selects 3. When the theoretical null is used for the same parameters, Procedure 2 selects 417,000 strategies, while GR selects 352,000. This is understandable as Efron's method relaxes the null specification, allowing more of the data to fall under the null distribution. This ultimately restricts the number of strategies that can be selectable by any procedure.

Finally, compared to procedures based on entirely different assumptions, the number of selected strategies by Procedure 2 is between RSW, and SC. RSW aims to control FDP at the same levels of $\gamma$ and $\alpha$. Still, it is very conservative while SC controls FDR at a level equal to $\alpha + (1 - \alpha) \ \gamma$. The comparison reinforces the cautionary message about applying a multiple comparisons procedure to a vast number of tests when all the information in the data is not taken into account, resulting in questionable answers. Learning the parameters of the data-generating process is, therefore, a valuable effort, especially in the context of the procedure that we propose in the paper, which relies on the \emph{local false discovery rate} as one of its primary inputs.

\section{Discussion}\label{sec:discuss}

Controlling the FDX provides an instrumental framework for applications where experiments are carried out only once or a few times. Unlike FDR methods, which only offer a long-run average guarantee, FDX methods provide a high-confidence control of the FDP for individual experiments. Under an empirical Bayes framework, we employ the PBD and its connection to \emph{lfdr} to develop an $(\gamma, \alpha)$-level FDX procedure and demonstrate its superiority over existing methods.
There are several open issues as the scope of this work progresses. 

One such issue is that several other error rates could be more attractive to specific researchers. For example, maximizing power only on the ``nicer'' realizations, that is, on those realizations where the FDP is indeed controlled at $\gamma$. Such exciting and highly relevant error rate notions are  left for future explorations.

Also, an essential question for FDX and FDR methods based on the empirical Bayes framework is how to estimate the \emph{lfdr} statistic in practice. This paper only considers a few available estimates to provide an implementable FDX procedure. Although valuable, a careful study of  \emph{lfdr} estimation is outside the scope of this work.


%

\section*{Acknowledgements}

We thank the two anonymous reviewers, Associate Editor and the Editor, whose comments have significantly improved the work and its exposition. The views expressed in this paper do not necessarily reflect those of the Federal Reserve Bank of Dallas or its staff. Pallavi Basu is grateful to Profs.~Sebastian D\"{o}hler, Etienne Roquain, and Daniel Yekutieli for insightful discussions and to Mr.~Gunashekhar Nandiboyina for parallel computing support. Part of the work was done while PB was at Tel Aviv University and has received funding from the European Research Council under the European Community's Seventh Framework Programme (FP7/2007-2013) / ERC grant agreement n° [294519]-PSARPS. Her research is partially supported by a MATRICS grant (MTR/2022/000073) from DST-SERB, India. All errors are our own.

{\small
\singlespace
\bibliographystyle{chicago}
\bibliography{refs}
}

\newpage

\appendix

\bigskip
\begin{center}
{\large\bf SUPPLEMENTARY MATERIAL}
\end{center}

\setcounter{section}{0}

\setcounter{page}{1}

We provide the proofs of Propositions 2 and 3, an illustration, and a counterexample. {\blue And we provide asymptotic guarantees of the proposed procedure.} Further, we offer a detailed numerical experiment replicating the data generating model for the stock returns trading strategies analyzed in Section 6. All numeric sections here refer to the content in the main paper.

\section{Proof of Proposition 2}\label{sec:supp_proof_prop2}

\begin{proof} We argue by contradiction. Suppose Proposition 2 is not true. Then for any claimed ``optimal'' decision rule there must exist a subset in the sample space $\mathbb{Z}$, which depends on the decision rule, such that $\PP\{w: \Z(w) \in \mathbb{Z}\} > 0$, where the decisions obtained by ranking \emph{lfdr} as in Procedure 2 are not preserved. {\blue That is, for the  claimed decision rule $\pmb \delta^*$,  there exist hypotheses $j_{\Z}$ and $\ell_{\Z}$, depending on the observed sample point or data $\Z$, such that $lfdr_{j_{\Z}} < lfdr_{\ell_{\Z}}$ but $\delta_{j_{\Z}} = 0$ and $\delta_{\ell_{\Z}} = 1$. For ease of notation, we skip the subscript $\Z$ when we further refer to indices $j$ and $\ell$, however, they are always contingent on the realized data.}

Then consider constructing  $\pmb \delta^{new}$, an alternative decision rule where we pick a particular pair $j$ and $\ell$ to modify so that  $\delta^{new}_j = 1$ and $\delta^{new}_{\ell} = 0$, and all other decisions are the same as in those in $\pmb \delta^*$.

 

Then we have,
\begin{align*}
\EE \left[ \sum_i (1-\theta_i) \delta_i^* \right]=& \EE_{\Z} \left[ \EE_{\theta|\Z}  \left\{ \sum_i (1-\theta_i) \delta_i^* \right\} \right]  \\
=& \EE_{\Z} \left[\sum_i (1-lfdr_i) \delta_i^* \{ I_{\Z \in \mathbb{Z}} + I_{\Z \in \mathbb{Z}^c} \}  \right]  \\
\leq &  \EE_{\Z} \left[ \sum_i (1-lfdr_i) \delta_i^{new} \{ I_{\Z \in \mathbb{Z}} + I_{\Z \in \mathbb{Z}^c} \} \right] \\
\leq & \EE \left[ \sum_i (1-\theta_i) \delta_i^{new} \right]. \label{opt_power_2}
\end{align*}


Next we verify that $\pmb \delta^{new}$ is a valid decision rule for FDX control. Note that the number of rejections, $k_0$, is the same in $\pmb \delta^*$ and $\pmb \delta^{new}$. Define $\pmb K$ as the set of rejected hypotheses common to both decision rules. Suppose we have shown that
\begin{equation}\label{toshow}
\PP_{\theta|\Z} \left (\sum_i \{(1-\theta_i)\delta^{new}_i\} > \gamma k_0\right ) \leq \PP_{\theta|\Z} \left (\sum_i \{(1-\theta_i)\delta^{*}_i\} > \gamma k_0\right )
\end{equation}
then if $ \pmb \delta^*$ is valid so is $\pmb \delta^{new}$. 

We now show that (\ref{toshow}) holds. Conditional on  $\pmb  Z$, we have that 
$$
\sum_i \{(1-\theta_i)\delta^{*}_i\} \sim \PBD\left(\{lfdr_{i \in \pmb K}, lfdr_{\ell}\}, k_0\right).
$$ 
Similarly, we have that $\sum_i \{(1-\theta_i)\delta^{new}_i\} \sim \PBD(\{lfdr_{i \in \pmb K}, lfdr_j\}, k_0)$. Recall that we have $lfdr_j < lfdr_{\ell}$. 
We consider random variables 
$$Y^{*} \sim \PBD(\{p_1 \cdots p_\ell^* \cdots p_{k_0}\}, k_0) \mbox{ and } 
Y^{new} \sim \PBD(\{p_1 \cdots p_j^{new} \cdots p_{k_0}\}, k_0)
$$ 
with $p_j^{new} < p_\ell^*.$ 
Then we only need to show that $\PP(Y^{new} > \gamma k_0) \leq \PP(Y^{*} > \gamma k_0)$.

Because the states of the hypotheses are independent, we can write 
$$\PP(Y^{*} > \gamma k_0) = p_\ell^* \PP(\bar{Y}^{*} > \gamma k_0 - 1) + (1 - p_\ell^*)\PP(\bar{Y}^{*} > \gamma k_0),
$$ where $\bar{Y}^* \sim \PBD(\{p_{i \in \pmb K}\}, k_0 - 1)$. Then,
\begin{align*}
\PP(Y^{*} > \gamma k_0) =& \PP(\bar{Y}^{*} > \gamma k_0) + p_\ell^* \{\PP(\bar{Y}^{*} > \gamma k_0 - 1) - \PP(\bar{Y}^{*} > \gamma k_0)\} \\
=& \PP(\bar{Y}^{*} > \gamma k_0) + p_\ell^* \PP(\bar{Y}^{*}  = \lceil \gamma k_0 - 1 \rceil) \\
\geq& \PP(\bar{Y}^{*} > \gamma k_0) + p_j^{new} \PP(\bar{Y}^{*}  = \lceil \gamma k_0 - 1 \rceil) \\
=& \PP(Y^{new} > \gamma k_0),
\end{align*}
where $\lceil x \rceil$ refers to the nearest integer strictly larger than $x$. This completes the proof. 
\end{proof}

\section{Illustration of the Computational Advantage of Procedure \ref{proc:proposedfaster}}

We implement Procedures~\ref{proc:proposedOR} and \ref{proc:proposedfaster} on three samples, each with $m=10,\!000$ independent $z$-values generated from the mixture distribution \eqref{mixF}:
$$
Z_i\sim (1-\pi)\mathcal{N}(0,1)+\pi\mathcal{N}(-2,1), \quad i\in [m],
$$
where $\pi=0.1$. Table \ref{Tbl.speed} reports the computation time for each run, as well as the progressive upper limits of rejections, at the FDX level ($\gamma = 0.1$, $\alpha = 0.05$).

We can see that the two procedures, which are shown to be equivalent in Section \ref{sec:computational-shortcuts}, reject the same number of hypotheses and thus produce the same FDP. However, the computation times required by the two procedures are very different:~Procedures~\ref{proc:proposedOR} requires 3.8 minutes whereas Procedures~\ref{proc:proposedfaster} accomplishes the same task in a fraction of a second. 
Our results also reveal the mechanics of Procedure 2: from Panel B, we see that it first quickly reduces the number of tests from $10,\!000$ to about 300 trials, represented by the entries in column $K_1$. Then, further reductions in the search are achieved, with the respective numbers of tests given by the entries in columns $K_2$ and $K$. The realized FDPs are close to the desired 0.1, with one in the three reported runs slightly exceeding 0.1.

\begin{boxedtable}{Computational advantages of Procedure 2}{Tbl.speed}
\tablexplain{The table report illustrative results of the computational advantage of Procedure 2 (Panel B) relative to Procedure 1 (Panel A). The underlying data generating process is a mixture model where 90\% come from a $\mathcal{N}(0,1)$ and 10\% come from a $\mathcal{N}(-2,1)$. We report the number of rejected nulls at each step of the oracle procedure and the total execution time for a sample number of runs (3 reported). FDX control is implemented at $\gamma = 0.1$ and with a confidence level $1-\alpha = 0.95$. We report the execution times for the experimental runs for a CPU using a 3GHz processor with 8GB RAM.}

\begin{ctabular}{c r r r c c} 
\midrule 
& $K_1$ & $K_2$ & $K$ &  Realized FDP & CPU Time\\ 
\cmidrule{2-6}
&\multicolumn{5}{c}{Panel A: Procedure 1}\\
\cmidrule{2-6}
Run 1 & - & - & 150 & 0.10 & 3.83 mins\\ 
Run 2 & - & - & 142 & 0.06 & 3.68 mins\\ 
Run 3 & - & - & 130 & 0.07 & 3.84 mins\\ 
\midrule
&\multicolumn{5}{c}{Panel B: Procedure 2}\\
\cmidrule{2-6}
Run 1 & 337 & 200  & 150 & 0.10 & 0.16 secs\\ 
Run 2 & 340 & 193  & 142 & 0.06 & 0.17 secs\\ 
Run 3 & 294 & 182  & 130 & 0.07 & 0.22 secs\\  
\midrule
\end{ctabular}
\end{boxedtable}

\section{Proof of Proposition 3}\label{sec:supp_proof_prop3}

\begin{proof} Note that in the exchangeable case, $\mu$ takes a specific sign (and value), either positive or negative, for all tests. Suppose without loss of generality that $\mu < 0$.  We show that, in such a case, ranking by \emph{lfdr} is equivalent to ranking by ascending values of z-scores, which is the same as ranking by ascending values of marginal \emph{lfdr}. Rewrite:
\[
P(\Z| \pmb \theta) \propto \exp \left \{ - \frac{1}{2} \Z \Sigma^{-1} \Z \right \} \cdot \exp \left \{ - \frac{1}{2} \mu^2 \pmb \theta \Sigma^{-1} \pmb \theta \right \} \cdot \exp \left \{ \mu \pmb \theta \Sigma^{-1} \Z \right \}.
\]  
If $z_1 > z_2$, and not altering other z-scores, then $$P(\Z| \theta_1 = 0, \theta_2 = 1, \pmb \theta_{\{-1, -2 \}}) > P(\Z| \theta_1 = 1, \theta_2 = 0, \pmb \theta_{\{-1, -2\}}),$$ where $\pmb{\theta}_{\{-1, -2\}}$ represents the vector of all $\theta$ that are not $\theta_1$ and $\theta_2$.  Note, in fact, that 
\begin{equation} \label{for_remark}
\ln \{ P(\Z| \theta_1 = 0, \theta_2 = 1, \pmb \theta_{\{-1, -2 \}})/P(\Z| \theta_1 = 1, \theta_2 = 0, \pmb \theta_{\{-1, -2\}}) \} \propto \mu (z_1 - z_2) ( \Sigma^{-1}_{ij} - \Sigma^{-1}_{ii} ),
\end{equation}
where $i \neq j$ represent all the off-diagonal entries; $\Sigma^{-1}_{ij} - \Sigma^{-1}_{ii} = 1/(\sigma^2 \cdot (\rho - 1)) < 0$; and $\mu < 0$. Denote by $\sigma^2 > 0$ and $\rho < 1$ the common variance and correlation respectively for the Gaussian multivariate density. Then to show that ranking by \emph{lfdr} is equivalent to ranking by increasing values of z-scores we must establish that ~$\PP(\theta_1 = 0|\Z) > \PP(\theta_2 = 0|\Z)$. The result follows since 
\begin{align*}
&\PP(\theta_1 = 0|\Z) - \PP(\theta_2 = 0|\Z) \\
&= \sum_{\pmb \theta_{\{-1, -2 \}} \in \{0, 1\}^{m-2}}\PP(\theta_1 = 0, \theta_2 = 1, \pmb{\theta}_{\{-1, -2 \}}|\Z) - P(\theta_1 = 1, \theta_2 = 0, \pmb{\theta}_{\{-1, -2 \}}|\Z),
\end{align*}
and $\theta_i \sim Ber(\pi)$ independently. 

Showing that ranking by ascending  $z$-scores is the same as ranking by ascending values of marginal \emph{lfdr} is simpler:~if $\mu < 0$, a higher positive value of a $z$-score indicates a higher marginal \emph{lfdr} value. Similarly, the converse holds when $\mu > 0$.

It remains to be shown that ranking by the \emph{lfdr} maximizes the power among all decisions that control the $FDX (\gamma, \alpha)$ in the exchangeable and Gaussian framework. To do so, if the hypotheses are not already ranked by the marginal \emph{lfdr}, consider switching decisions. As we show in the proof to Proposition~\ref{Prop:OR}, swapping decisions leads to an increase in power. Thus,  ranking by marginal \emph{lfdr}, will lead to maximum power. Moreover, note that when $\mu < 0$ and $z_1 > z_2$ we have $P(\theta_1 = 0, \theta_2 = 1, \pmb \theta_{\{-1, -2 \}}|\Z) > P(\theta_1 = 1, \theta_2 = 0, \pmb \theta_{\{-1, -2 \}}|\Z)$ (as shown in the previous paragraph). Hence, if the original procedure had valid FDX control, so does the modified procedure.
\end{proof}

\begin{remark}
For generalized non-Gaussian distributions, (\ref{for_remark}) highlights the sufficient condition to guarantee both the optimality of ranking and that ranking by the marginal \emph{lfdr} is adequate. This condition is naturally satisfied for exchangeable elliptical distributions such as the multivariate normal, t, or Laplace distributions. Similarly, for a general test statistic, such as the absolute values of the z-score or of the t-statistic, it is sufficient that the left side term of (\ref{for_remark}) is negative (or positive in $T(\cdot)$). That is, the sufficient condition is that if $T_j(\z) := |z_j| < |z_i| =: T_i(\z)$, then $P(T(\Z)| \theta_i = 0, \theta_j = 1, \pmb \theta_{\{-i, -j \}}) < P(T(\Z)| \theta_i = 1, \theta_j = 0, \pmb \theta_{\{-i, -j\}})$ holds for the joint density. When iid, this condition factorizes to the monotone likelihood ratio (MLR) criterion. Thus this condition may be viewed as a generalized MLR condition for multivariate densities.
\end{remark}

{\blue \section{Asymptotic Guarantees} \label{Sec.Asymptotic}}

In this section, we establish the asymptotic validity and the asymptotic ranking. The framework, or definitions of validity and optimality, are developed along the proofs. The results are shown under the mildest conditions, the convergence of the lfdr estimates, without requiring stringent assumptions.

\begin{proposition} (Asymptotic Validity.) To establish that $\PP \left(\mbox{FDP} \geq \gamma + \epsilon \right)\leq\alpha + o(1)$, for any $\epsilon > 0$, under Assumption 4.3 of Basu et al.~(2018).
\end{proposition}

\begin{proof}

Recall that, by design, the oracle Procedure 1 ensures
\[
\PP_{\theta|\Z} \left\{\sum_{i\in\mathcal R_K^{OR}} (1-\theta_i) > \gamma K\right\} \leq \alpha. 
\]
Note that the oracle lfdr statistic $(T^{OR}_i)_{i\in[m]}$ are estimated by the data-driven test statistics $(T^{DD}_i)_{i\in[m]}$. And our data-driven procedure ensures $\PP_{\widehat{\theta|\Z}} \left\{\sum_{i\in\mathcal R_K} (1-\theta_i) > \gamma K\right\} \leq \alpha$, where note that $K$ is determined by the data and we avoid the superscript `DD' for notational simplicity. We show that $\PP \left(\mbox{FDP} \geq \gamma + \epsilon \right)\leq\alpha + o(1)$ for any $\epsilon > 0$.

First, we take a brief excursion. Consider $(Y_i, Y_i')$ as Bernoulli random variables with $P(Y_i = 0) = q_i$ and $P(Y_i' = 0) = q_i'$. Let $X_K = \sum_{i\in\mathcal R_K} (1-Y_i)$ and $X_K' = \sum_{i\in\mathcal R_K} (1-Y_i')$. Note that the index set $R_K$ may, and in our case will, depend on the $q_i$ or $q_i'$ and the probabilities below are conditional on $q_i$ and $q_i'$. Then,
%
\begin{align*}
P(X_K/K \geq \gamma + \epsilon) & = P(X_K/K \geq \gamma + \epsilon, X_K'/K > \gamma)  + P(X_K/K \geq \gamma +\epsilon, X_K'/K \leq \gamma) \\
& \leq P(X_K' > \gamma K)  + P(X_K/K -  X_K'/K \geq \epsilon), 
\end{align*}
where the joint probabilities are defined on the coupled events retaining the marginal distributions. That is, without loss of generality, let the joint distribution be $(Y_i = 0,  Y_i' = 0)$ with probability $q_i$, $(Y_i = 1,  Y_i' = 0)$ with probability $q_i' - q_i$ when $q_i' \geq q_i$, $(Y_i = 0,  Y_i' = 1)$ with probability $0$, and $(Y_i = 1,  Y_i' = 1)$ with probability $1-q_i'$. 

Now we analyze the scaled differences, for the $\epsilon > 0$,
\begin{align*}
P(X_K/K -  X_K'/K \geq \epsilon) =& P((\sum_{i\in\mathcal R_K} 1)^{-1} \left \{ \sum_{i\in\mathcal R_K} (1-Y_i) - \sum_{i\in\mathcal R_K} (1-Y_i') \right \} \geq \epsilon)\\
=& P((\sum_{i\in\mathcal R_K} 1)^{-1} \left \{ \sum_{i\in\mathcal R_K} (1-Y_i) - \sum_{i\in\mathcal R_K} (1-Y_i') \right \}_+ \geq \epsilon)\\
\leq& \epsilon^{-1}E((\sum_{i\in\mathcal R_K} 1)^{-1} \left \{ \sum_{i\in\mathcal R_K} (1-Y_i) - \sum_{i\in\mathcal R_K} (1-Y_i') \right \}_+)\\
\leq& \epsilon^{-1}(\sum_{i\in\mathcal R_K} 1)^{-1} E (\sum_{i\in\mathcal R_K} \left \{ (1-Y_i) - (1-Y_i') \right \}_+),
\end{align*}
where we use that for any random variable $\{X \geq \epsilon\} \Leftrightarrow \{X_+ \geq \epsilon\}$, thereby enabling us to apply Markov's inequality on the positive part. Thus we obtain that,
\begin{equation} \label{sample_prob_bound}
P(X_K/K -  X_K'/K \geq \epsilon) \leq \epsilon^{-1}(\sum_{i\in\mathcal R_K} 1)^{-1} \sum_{i\in\mathcal R_K}  \{ (q_i - q_i')_+ \}.
\end{equation}

Let's reconnect to your procedure. Here $\{Y_i = 0\} \equiv \{\theta_i = 0 | \Z\}$ and $q_i \equiv T^{OR}_i$ and $q_i' \equiv T^{DD}_i$. And the set $\mathcal R_K$ is the set of rejected hypotheses using the data-driven procedure. Note our procedure guarantees that for a random $K$, we have $P(X_K'/K > \gamma) \leq \alpha$. To complete the proof, we need to show $EP(X_K/K -  X_K'/K \geq \epsilon) = o(1)$ or we will show that $E(\epsilon^{-1}(\sum_{i\in\mathcal R_K} 1)^{-1} \sum_{i\in\mathcal R_K}  \{ (q_i - q_i')_+ \}) = o(1)$. Note that (\ref{sample_prob_bound}) is bounded by $\epsilon^{-1}$. Thus if we can show that $(\sum_{i\in\mathcal R_K} 1)^{-1} \sum_{i\in\mathcal R_K}  \{ (q_i - q_i')_+ \} = o_P(1)$ we will be done. Consider $m (\sum_{i\in\mathcal R_K} 1)^{-1}$, which is  
\[
[m^{-1} \sum_i \mathcal{I}_{T^{DD}_i \leq \lambda_{DD}}]^{-1} \leq [m^{-1} \sum_i \mathcal{I}_{T^{DD}_i \leq \epsilon'}]^{-1} = [E(T^{OR} \leq \epsilon') + o_p(1)]^{-1} = O_p(1),
\]
where $\epsilon' > 0$ and depends on $(\gamma, \alpha)$ as illustrated more shortly. Similarly
\[
m^{-1}\sum_{i\in\mathcal R_K}  \{ (T_i^{OR} - T_i^{DD})_+ \} \leq m^{-1} \sum_{i} |T_i^{OR} - T_i^{DD}| = o_p(1)
\]
similar to proofs in Basu et al.~(2018) in Appendix B.2.

We now illustrate how to obtain a lower bound depending on $(\gamma, \alpha)$ for our data-driven procedure. We formulate it in the size of $\mathcal R_m^{\epsilon'}$ where $\mathcal R_m^{\epsilon'} := \{i \in [m]: T^{DD}_i \leq \epsilon' \}$. We need to check that for all $k = 1, 2, 3, \cdots, m$ it holds that $P(X_k' > \gamma k) \leq \alpha$. To find an $\epsilon'$ note that 
\begin{equation} \label{bin_large_deviation}
P(X_k' > \gamma k) \leq P(X_k' \geq \gamma k) \leq P(Bin(k, \epsilon') \geq \gamma k) \leq \exp\{-kH\} \leq \exp\{-H\},
\end{equation}
where $H(\gamma, \epsilon') := \gamma \ln(\gamma/\epsilon') + (1-\gamma) \ln((1-\gamma)/(1-\epsilon'))$ denotes the `relative entropy' between an $\epsilon'$-coin and a $\gamma$-coin (See Theorem 1 of Arratia and Gordon, 1989) where $0 < \epsilon' < \gamma < 1$. Note that as $\epsilon' \rightarrow 0$, the relative entropy grows and thus the right-hand side of (\ref{bin_large_deviation}) converges to zero. Therefore, one can easily choose an $\epsilon'$ to ensure that $\exp\{-H\} \leq \alpha$ for any $\alpha > 0$. We also remark that the entropy bound is useful for larger $k$ and thus, for all practical purposes, a choice of $\epsilon' \approx \gamma * \alpha$ suffices as a lower bound for the threshold of our data-driven procedure.

\end{proof}

\begin{proposition} (Asymptotic Optimal Ranking.) Under assumption 4.3 of Basu et al.~(2018), the proposed procedure using the estimated lfdr has an asymptotically optimal ranking. \end{proposition}

\begin{proof} 

The proof is done in two parts. First, we show that selecting top `k' by the estimated lfdr test statistics isn't asymptotically suboptimal compared to the oracle lfdr test statistics. Next, we propose a relaxed data-driven procedure, in practice equivalent to the proposed procedure, which is shown to be asymptotically valid. And then, we offer that thresholding by the relaxed data-driven approach is asymptotically optimal. 

Consider ranked by either top-k local FDR or the estimated local FDR test statistics. Our target power is $ETP := E\sum_i(\theta_i \delta_i)$. Let $k(Z)$ potentially be dependent on the data be the top-k rejected hypotheses, wherein in one procedure, the oracle lfdr ranks the hypotheses. In another, the estimated lfdr statistic ranks them. Define the expected regret in scaled power as $m^{-1} E \{ \sum_i \theta_i \delta^{OR_k}_i - \sum_i \theta_i \delta^{DD_k}_i\}$ which conditional on the data reduces to $m^{-1} E \{ \sum_i (1-lfdr_i) \delta^{OR_k}_i - \sum_i (1-lfdr_i) \delta^{DD_k}_i\}$. Decomposing,
\begin{align*}&m^{-1}\text{Regret}(\delta^{OR_k}, \delta^{DD_k}) = m^{-1}\{ \sum_i (1-lfdr_i) \delta^{OR_k}_i - \sum_i (1-lfdr_i) \delta^{DD_k}_i\} \\&= m^{-1} \{ \sum_i (1-lfdr_i) \delta^{OR_k}_i - \sum_i (1-\widehat{lfdr_i}) \delta^{OR_k}_i\} + m^{-1} \{ \sum_i (1-\widehat{lfdr_i}) \delta^{OR_k}_i \\& - \sum_i (1-\widehat{lfdr_i}) \delta^{DD_k}_i\} + m^{-1} \{ \sum_i (1-\widehat{lfdr_i}) \delta^{DD_k}_i - \sum_i (1-lfdr_i) \delta^{DD_k}_i\}\\& \leq m^{-1} \{ \sum_i (1-lfdr_i) \delta^{OR_k}_i - \sum_i (1-\widehat{lfdr_i}) \delta^{OR_k}_i\} \\& + m^{-1} \{ \sum_i (1-\widehat{lfdr_i}) \delta^{DD_k}_i - \sum_i (1-lfdr_i) \delta^{DD_k}_i\} \\&=o_P(1),\end{align*}where $m^{-1} \{ \sum_i (1-\widehat{lfdr_i}) \delta^{OR_k}_i - \sum_i (1-\widehat{lfdr_i}) \delta^{DD_k}_i\}$ is non-positive as $\delta^{DD_k}_i$ is ranked by the lower to higher values of $\widehat{lfdr_i}$ and hence higher to lower values of $(1-\widehat{lfdr_i})$. 

The size of the rejected set is related to the thresholding statistic. In this regard, we would like to introduce the idea of a relaxed data-driven procedure. Choose a (deterministic) sequence of $\eta_m \rightarrow 0$ which isn't related to the data. We define a relaxed data-driven procedure as $K^{DDr} := \sup \{k \in [m]: \PP_{\widehat{\theta|Z}} \left(\mbox{FDP(k)} \geq \gamma + \epsilon \right)\leq\alpha + \eta_m \text{ for all } \epsilon > 0\}$. Towards the end of this proof, we establish that this relaxed data-driven procedure is also asymptotically valid.
We establish that the size of the rejected set by the relaxed data-driven procedure is, with high probability, higher than that chosen by the oracle procedure with a growing number of tests. In other words, we show that $K^{DDr} \geq K^{OR}$ with high probability and with increasing $m$. To establish this, we show that $K^{OR}$ satisfies the relaxed data-driven criterion with high probability. For any $\varepsilon > 0$, we show that $\PP_{\widehat{\theta|Z}} \left(\mbox{FDP($K^{OR}$)} \geq \gamma + \epsilon \right) \leq \alpha + \eta_{m}$ for $m \geq m_0$. 

Following the thoughts in the proof of Proposition 4, \begin{align*} &P(X_K/K \geq \gamma + \epsilon) = P(X_K/K \geq \gamma + \epsilon, X_K'/K > \gamma)  + P(X_K/K \geq \gamma +\epsilon, X_K'/K \leq \gamma) \\& \leq P(X_K' > \gamma K)  + P(X_K/K -  X_K'/K \geq \epsilon) \leq \alpha + P(X_K/K -  X_K'/K \geq \epsilon).\end{align*}
Further an $m_0$ is chosen such that for all $m \geq m_0$, we have that,
\[
\PP \{ P(X_K/K -  X_K'/K \geq \epsilon) > \eta_{m} \} < \varepsilon',
\]
for any choice of $\varepsilon' > 0$. Thus $P(K^{DDr} \geq K^{OR}) \geq 1 - \varepsilon'$. This is true if $E\{T_i^{OR} - T_i^{DD}\}/\eta_m \rightarrow 0.$ Thus, the regret for the relaxed data-driven procedure can be expressed as:\begin{align*}&m^{-1}\text{Regret}(\delta^{OR_k}, \delta^{DDr}) = m^{-1}\{ \sum_i (1-lfdr_i) \delta^{OR_k}_i - \sum_i (1-lfdr_i) \delta^{DDr}_i\} \\&\leq m^{-1}\{ \sum_i (1-lfdr_i) \delta^{OR_k}_i - \sum_i (1-lfdr_i) \delta^{DD_k}_i\} = m^{-1}\text{Regret}(\delta^{OR_k}, \delta^{DD_k}),\end{align*}whenever $K^{DDr} \geq K^{OR}$. Further note that if $K^{DDr} < K^{OR}$, the scaled regret is bounded by 1. Thus $m^{-1}E\text{Regret}(\delta^{OR_k}, \delta^{DDr}) = o(1)$. And further, as the relaxed data-driven procedure makes more rejections than the proposed data-driven procedure, and as argued in proposition 4, $[m^{-1} \sum_i \mathcal{I}_{T^{DD}_i \leq \lambda_{DDr}}]^{-1} \leq [m^{-1} \sum_i \mathcal{I}_{T^{DD}_i \leq \lambda_{DD}}]^{-1} = O_p(1)$. Thus we have established that $E \sum_i \theta_i \delta^{OR}_i = \{E \sum_i \theta_i \delta^{DDr}_i\}(1+o(1))$. This establishes the asymptotic optimality in ranking by the estimated local FDR statistic. 


We conclude the proof by establishing that the relaxed data-driven procedure is asymptotically valid, thereby enabling its theoretical uses. To see this, fix any $\varepsilon > 0$, 
\begin{align*}
&P(X_K/K \geq \gamma + \varepsilon) \\ & = P(X_K/K \geq \gamma + \varepsilon, X_K'/K \geq \gamma + \varepsilon/2)  + P(X_K/K \geq \gamma +\varepsilon, X_K'/K < \gamma + \varepsilon/2) \\
& \leq P(X_K'/K \geq \gamma + \varepsilon/2)  + P(X_K/K -  X_K'/K > \varepsilon/2) \\
& \leq \alpha + \eta_m + P(X_K/K -  X_K'/K > \varepsilon/2).
\end{align*}
Following the proof of Proposition 4, we have that $EP(X_K/K -  X_K'/K > \varepsilon/2) = o(1)$ under the usual condition that $ |T_i^{OR} - T_i^{DD}| = o_p(1)$. Thus $\PP \left(\mbox{FDP} \geq \gamma + \epsilon \right)\leq\alpha + o(1)$ thereby ensuring the validity of the relaxed data-driven procedure.
\end{proof}

\begin{remark}
Note that without strict assumptions, it is nontrivial to establish that the proposed and the relaxed data-driven procedures are equivalent. The power of the proposed method can be improved using the theoretical underpinnings of the relaxed data-driven approach. However, the relaxed data-driven procedure is practically less appealing. For our purpose of proof, our aim here is to establish the asymptotic optimality of the test statistic as a ranking statistic. This has been confirmed by slightly inflating the cut-off or the thresholding while keeping the ranking test statistic intact and, most importantly, without introducing additional assumptions.
\end{remark}

\section{A Counter Example Regarding Ranking Optimality} \label{Sec.Counterexample}

Using a numerical experiment, we demonstrate that the \lfdr \ ranking may not be optimal when the tests are not exchangeable. We follow the setting described in Section 5.2 of \cite{heller2021optimal} to provide a counterexample for the best ranking for FDX control. Ten z-scores are generated from the two-group model with $\theta_i \sim Bernoulli(0.3)$ independently. Further $\Z|\pmb \theta \sim N(-1.5*\pmb \theta, \Sigma + 0.01 * diag(\pmb \theta))$. Because the computational burden grows exponentially, we simulate only 10 test statistics and work with  $\gamma = 0.5$. Since for independent (or exchangeable) $z$-values, we obtain an optimal ranking (by sorting on \emph{lfdr}) regardless of the size of $\gamma$, the framework serves as a good counterexample. 

We work with a block diagonal matrix $\Sigma$ with two blocks, each being equicorrelated with a varied choice of $\rho$. We run  our experiment 200 times. Each time we consider ranking the indices by the \emph{lfdr}, which is the optimal ranking for exchangeable tests, and we select the indices with the lowest two values. We then divide the hypotheses into the two blocks in  $\Sigma$, which are known to the oracle. From the two blocks, we further select the top indices with the minimum \lfdr \ values. We find this to be the most intuitive way to counter the dependencies within blocks and potentially make a better rejection due to the positive dependence within the blocks. Table \ref{Tbl.Counterexample} reports the percentage of times that the \lfdr \ ranking provides higher tail probabilities than the ranking based on the \lfdr \ values within blocks separately. The values are reported for varied values of $\rho$ in percentage for 200 experimental runs. If we look at $\rho = 0.3$ in the table, 13\% of the time, selecting the top indices from the two blocks separately produces a lower value for the tail probability than selecting the top indices from all the indices combined. These differences in the tail probabilities can be high. For example, consider $\rho = 0.5$:~the tail probability is higher by an average of 16.7\%. This illustrates that the \lfdr \ ranking can be suboptimal when the exchangeability assumption fails to hold.

\begin{boxedtable}{Counterexample to the optimal \emph{lfdr} ranking}{Tbl.Counterexample}
\tablexplain{The table reports the proportion of time the tail probability was lower for an alternative ranking than for the \lfdr \ ranking in Procedure 2. The \lfdr \ ranking may be substantially improved for moderately correlated test statistics by incorporating the correlation structure into account. This demonstrates that the \lfdr \ ranking and thresholding procedure may not be optimal for non-exchangeable tests where the correlation structure is unknown to a practitioner.}
\begin{ctabular}{c c c c c c c c} 
\midrule 
$\rho $ & 0.01 & 0.1 & 0.3 & 0.5 & 0.7 & 0.9 \\ [0.5ex] 
\midrule 
Contradictory tail probabilities (in \%) & 0 & 4.5 & 13 & 16 & 21.5 & 31 \\ 
\midrule 
\end{ctabular}
\end{boxedtable}

{\blue \section{Completed Settings for Numerical Experiments}}

In this section, we present some additional settings for different choices of parameters for the numerical experiment setting discussed in Section 5.1. First, we demonstrate a scenario when $\pi = 0$ and discuss the findings.

\begin{boxedtable}{[Table 1 for $\pi = 0$] Comparison of different procedures}{Tbl.Comparison.2}
\tablexplain{The table compares the performance of Procedure 2 relative to some popular methods:~SC, BH, GR, and LR. Three versions of Procedure 2 are implemented:~the oracle version (Oracle), a version where the parameters of the data generating process are estimated from the data ($lfdr$), and a version where we impose the assumption that $\pi = 0$, $lfdr(\hat{\pi} = 0)$. The data-generating process is a mixture model, where with probability $1-\pi = 1$ the test is drawn from a $\mathcal{N}(0,1)$ (null) where $\pi = 0$. The simulation considers 5000 tests. We repeat the exercise for 1000 simulations. FDX control is implemented at $\gamma = 0.05$ and with a confidence level $1-\alpha = 0.95$. FDR control is implemented at a nominal level of $\alpha = 0.05$. } 

\begin{ctabular}{l l rrrrrrr} 
\midrule
&&&&&&\multicolumn{3}{c}{Procedure 2}\\
\cmidrule{7-9}
& & SC & BH &  GR & LR & Oracle & $lfdr$ & $lfdr(\hat{\pi} = 0)$\\  
\midrule
$\pi = 0$ & FDX & 0.05 & 0.044 & 0.069 & 0.068 & 0 & 0.05 & 0.049\\ 
& FDR & 0.05 & 0.044 & 0.069 & 0.068 & 0 & 0.05 & 0.049\\ 
& Power (\%) & - & - & - & - & - & - & -\\ 
\midrule 
\end{ctabular}
\end{boxedtable}

All methods control the FDR/FDX at nearly 0.05. Note that since no tests come from the alternative (i.e., $\pi=0$), then any rejection that happens is false, and hence if there are rejections, the FDP is precisely 1. If there are no rejections, the FDP is zero by definition. It follows that the FDX (probability of FDP being above the threshold) equals the FDR. Further, for the oracle procedure, the lfdr is identically one for all tests. Thus, the procedure cannot make any null rejections, resulting in the FDP being identically zero, which makes both FDX and FDR equal to zero.

Next, we complete the choices of parameters that still need to be reported in Table 1. All the primary insights remain the same, as discussed in Section 5.1.

\begin{boxedtable}{[Expanded Table 1] Comparison of different procedures}{Tbl.Comparison.3}
\tablexplain{The table compares the performance of Procedure 2 relative to some popular methods:~SC, BH, GR, and LR. Three versions of Procedure 2 are implemented:~the oracle version (Oracle), a version where the parameters of the data generating process are estimated from the data ($lfdr$), and a version where we impose the assumption that $\pi = 0$, $lfdr(\hat{\pi} = 0)$. The data-generating process is a mixture model, where with probability $1-\pi$ the test is drawn from a $\mathcal{N}(0,1)$ (null), and with probability $\pi$ the test is drawn from $\mathcal{N}(\mu, 1)$ (alternative), where $\pi \in \{0.1, 0.3\}$ and $\mu \in \{-1.5, -2.5\}$. Each simulation considers 5000 tests. We repeat the exercise for 1000 simulations. FDX control is implemented at $\gamma = 0.05$ and with a confidence level $1-\alpha = 0.95$. FDR control is implemented at a nominal level of $\alpha = 0.05$. } 
\begin{ctabular}{l l rrrrrrr} 
\midrule
&&&&&&\multicolumn{3}{c}{Procedure 2}\\
\cmidrule{7-9}
& & SC & BH &  GR & LR & Oracle & $lfdr$ & $lfdr(\hat{\pi} = 0)$\\  
\midrule
$\pi = 0.1$ & FDX & 0.312 & 0.262 & 0.06 & 0.059 & 0.051 & 0.09 & 0.085\\ 
$\mu = -1.5$& FDR & 0.061 & 0.06 & 0.035 & 0.035 & 0.032 & 0.046 & 0.043\\ 
& Power (\%) & 1.1 & 1 & 0.3 & 0.3 & 0.3 & 0.4 & 0.3\\ 
\midrule 
$\pi = 0.1$ & FDX & 0.489 & 0.366 & 0.035 & 0.02 & 0.041 & 0.045 & 0.031\\ 
$\mu = -2.5$& FDR & 0.051 & 0.045 & 0.013 & 0.003 & 0.027 & 0.027 & 0.024\\ 
& Power (\%) & 35.9 & 33.9 & 13.6 & 4.6 & 25.2 & 25.5 & 23.6\\  
\midrule 
 $\pi = 0.3$& FDX & 0.470 & 0.243 & 0.037 & 0.036 & 0.056 & 0.071 & 0.05\\ 
$\mu = -1.5$& FDR & 0.049 & 0.036 & 0.01 & 0.009 & 0.015 & 0.018 & 0.013\\  
& Power (\%) & 7.7 & 4.6 & 0.3 & 0.3 & 1 & 1.4 & 0.6\\ 
\midrule
$\pi = 0.3$& FDX & 0.453 & 0.008 & 0.002 & 0.000 & 0.048 & 0.048 & 0.000\\ 
$\mu = -2.5$& FDR & 0.049 & 0.035 & 0.031 & 0.002 & 0.039 & 0.039 & 0.027\\  
& Power (\%) & 61.6 & 54.4 & 52.3 & 10.3 & 56.8 & 56.6 & 48.8\\ 
\midrule
\end{ctabular}
\end{boxedtable}

\section{Simulation of Stock Returns Trading Strategies}\label{Sec.FinSimulation}

In this section we present a simulated version of the empirical application presented in Section~\ref{Sec.Application} that follows the set up of \cite{chordia2020anomalies}. A total of $N = 2,\!000$ stock returns are generated for $T=500$ periods. Returns follow a linear factor structure 
\[ ret_{it} = \alpha_i + \beta_i' F_t + \varepsilon_{it}, \]
where the factors, $F$, are drawn from a multivariate normal distribution with the mean and the covariance matrix matching closely those of the five \cite{fama2015five} factor model, augmented with \cite{carhart1997persistence} momentum. For each stock, $\alpha$ represents the return that an investor could realize in excess of the risk generated by the factors, and is drawn from a $\mathcal{N}(0, \sigma_{\alpha}^2)$. The idiosyncratic noise follows $\varepsilon_{it} \sim \mathcal{N}(0, \sigma_{\varepsilon}^2)$ with $\sigma_{\varepsilon} = 15.1\%$. In each time period (i.e., month), we draw $S = 5,000$ trading signals for each stock:~a fraction $\pi$ of the trading signals are informative (although imperfectly so) about the $\alpha$ of each stock: $s_{it} = \alpha_i + \eta_{it}$, where $\eta_{it} \sim \mathcal{N}(0, \sigma_{\eta})$.  A fraction $1-\pi$ contain just noise:~$s_{it} = \eta_{it}$. Informative and uninformative trading signals might share some common noise through the correlation coefficient $\rho_\eta$ among signals.

A stock trader constructs $5,\!000$ trading strategies by sorting stocks at the end of each month based on the signals realizations. She forms 10 portfolios, buys the portfolio corresponding to the largest signals, shorts the portfolio corresponding to the lowest, and holds the position until the end of the month when she repeats the sorting procedure and forms new portfolios. The performance of these $5,\!000$ long-short portfolios is evaluated by regressing the time series of 500 portfolio return observations on the realizations of the factors. A $t$-statistics on the estimated regression intercept (i.e., the portfolio alpha or abnormal return) is used to evaluate each strategy in classical hypothesis testing, and altogether in multiple hypothesis testing. The entire simulation procedure is repeated $1,\!000$ times.

Figure~\ref{Fig.SimHist} provides a visual representation of the null and alternative distributions in scenarios that differ in the strength of the informative signal, $\sigma_\alpha \in \{1.1\%, 1.25\%, 1.5\% \}$. As the signal becomes stronger, the truly informative trading signals are more easily identifiable. However, even the most favorable scenario is rather complicated for a multiple hypothesis testing procedure as there is substantial overlap between null and alternative distributions.  

\begin{boxedfigure}{Simulation scenarios}{Fig.SimHist}
\tablexplain{The figure presents a visual comparison of different distributions of test statistics for the trading strategy simulation.  In the base case situation:~$\pi = 0.1$,  $\sigma_{\varepsilon} = 0.15$, $\sigma_\eta = 0.2$, and $\rho_\eta = 0$. The figure presents different simulations that vary $\sigma_\alpha \in \{1.1\%, 1.25\%, 1.5\% \}$.}
\begin{center}
\includegraphics[width=0.8\textwidth]{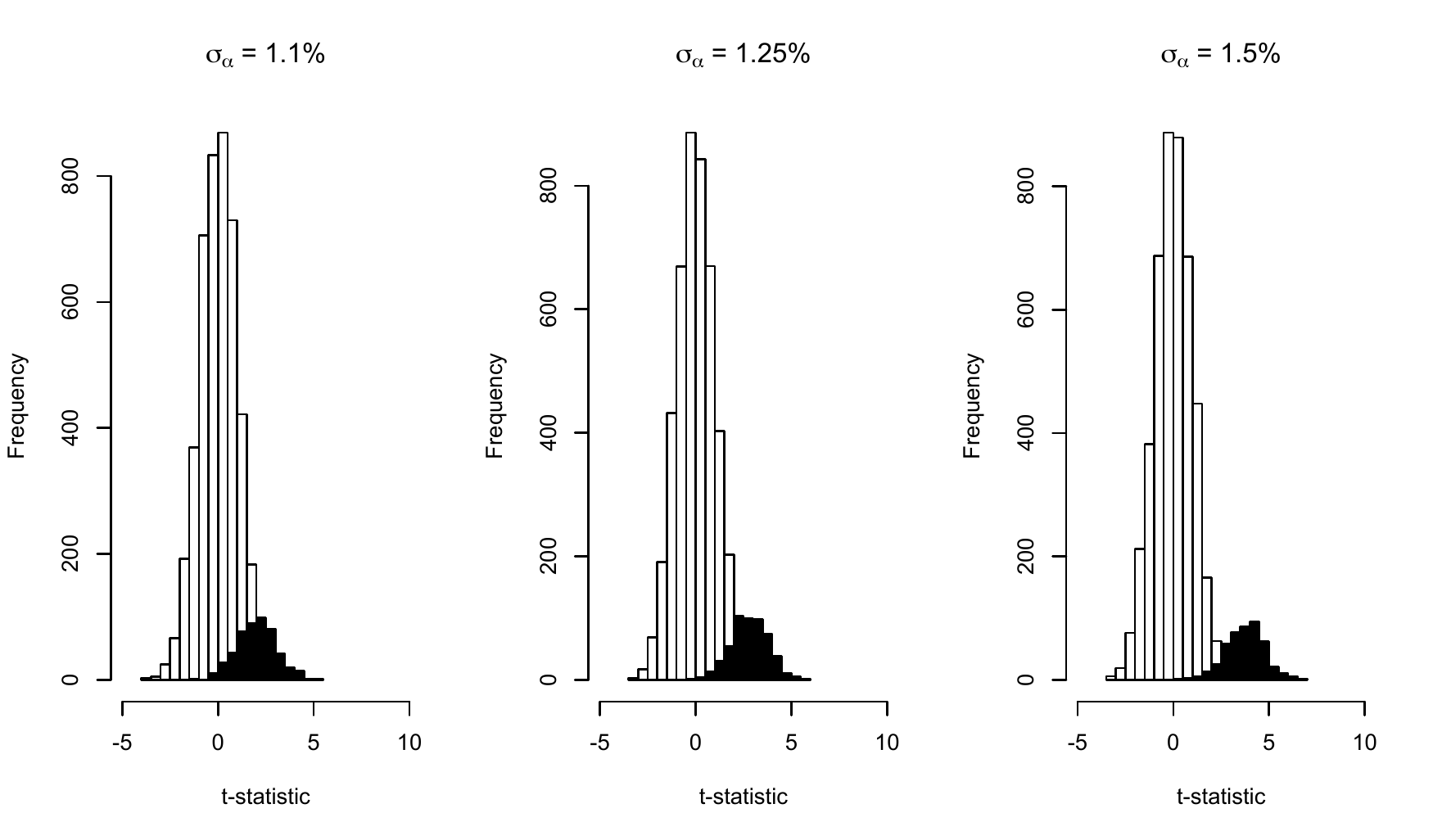}
\end{center}
\end{boxedfigure}

We analyze the simulated data under two scenarios that differ in whether the econometrician knows the true proportion of nulls and the parameters of the null distribution (Oracle) or has to estimate that from the data (Data-Driven). {\blue Here, the exact null distribution is unknown or complicated due to the complex data-generating process.} Thus, we must construct the Oracle in a non-traditional way:~we allow the Oracle to estimate the parameters of the null distribution by temporarily endowing it with the knowledge of which tests are null and which are not. The Oracle can then estimate the null and alternative parameters using this information. After that, it forgets knowing which tests are null and tries to uncover that truth by running the various procedures. In the Data-Driven scenario, the econometrician does not know the structural parameters and estimates them all. 

We report the results in Table~\ref{Tbl.FinanceSimulation} and Table~\ref{Tbl.FinanceSimulation2}, respectively.  Both tables compare the average FDX, FDR, and power of Procedure 2 to those of  \cite{SunCai07} (SC), \cite{BenHoc95} (BH), \cite{GuoRomano07} (GR), and \cite{LehmannRomano05} (LR). Both tables contain results under different specifications of the standard deviation of $\alpha$, $\sigma_\alpha \in \{1.1\%, 1.25\%, 1.5\% \}$ and under varying specifications of the pairwise correlation coefficient for trading signals, $\rho_\eta \in \{0, 0.1, 0.2\}$. The  $\sigma_\alpha$  determines the signal to noise ratio in the trading signal. Hence, a higher $\sigma_\alpha$ makes the signals more informative. Similarly, a non-zero correlation makes the signal to noise ratio higher by reducing the background noise, but it also places procedures outside of the canonical i.i.d case.

In almost every scenario considered, Procedure 2 sits between BH and GR in terms of FDX and power. From Table~\ref{Tbl.FinanceSimulation} we see that the Oracle version is able to maintain, although a bit conservatively, FDX control while delivering reasonable power, especially in the very difficult scenario (i.e.,  $\sigma_\alpha = 1.1\%$). Procedure 2  outperforms GR, its most natural comparison, in all but one scenario:~when $\sigma_\alpha = 1.1\%$ GR delivers a slightly higher FDX.
Although it is designed to control FDR, as opposed to FDX, the most applied multiple testing procedure is BH. Relative to Procedure 2, BH does guarantee FDR control while maintaining a higher power but at the cost of relatively large FDX, often over 30\%. Increasing correlation among tests makes it easier to separate null and alternative observations, thus leading to an increase in power for all procedures, but even in this case of simple correlation structure which maintains exchangeability, Procedure 2 is still able to accomplish its objectives of delivering FDX control and outpacing the power of other FDX  procedures.

\begin{boxedtable}{Stock return simulation (Oracle)}{Tbl.FinanceSimulation}
\tablexplain{The table compares across 5000 trading strategies the average FDX, FDR, and power obtained from Procedure 2 to those of \cite{SunCai07} (SC), \cite{BenHoc95} (BH), \cite{GuoRomano07} (GR), and \cite{LehmannRomano05} (LR). In each simulation, the returns of $2,\!000$ stocks for 500 observations is generated from a linear factor structure model,  $ret_{it} = \alpha_i + \beta_i' F_t + \varepsilon_{it}$, where $\alpha_i$ represents the return that an investor could realize in excess of the risk generated by the factors, $F_t$ and is drawn from a $\mathcal{N}(0, \sigma_{\alpha}^2)$ and $\varepsilon_{it} \sim \mathcal{N}(0, \sigma_{\varepsilon}^2)$.  In each time period (i.e., month), we draw $S = 5,\!000$ trading strategies for each stock:~a fraction $\pi$ of the trading signals are informative (although imperfectly) about the $\alpha$ of each stock:~$s_{it} = \alpha_i + \eta_{it}$, where $\eta_{it} \sim \mathcal{N}(0, \sigma_{\eta})$.  A fraction $1-\pi$ contain just noise: $s_{it} = \eta_{it}$. Informative and uninformative trading signals can  share some common noise through the correlation coefficient $\rho_\eta$. Each month's stocks are sorted into deciles based on the trading signal realization, and a long short portfolio is obtained from buying stocks in the top deciles and shorting stocks in the bottom deciles. A $t$-statistics of the portfolio regression alpha serves as the relevant test statistic in evaluating each trading strategy/long-short portfolio.  In the base case situation:~$\pi = 0.1$,  $\sigma_{\varepsilon} = 0.15$, $\sigma_\eta = 0.2$. The table presents different simulations that vary $\sigma_\alpha \in \{1.1\%, 1.25\%, 1.5\% \}$ and $\rho_\eta \in \{0, 0.1, 0.2\}$. Each procedure is implemented as an Oracle, that has knowledge of, or is able to accurately estimate, the parameters of the data-generating process. FDX control is implemented at $\gamma = 0.05$ and with a confidence level $1-\alpha = 0.95$. FDR control is implemented at a nominal level of $\alpha = 0.05$.}

\begin{ctabular}{l l rrrrr} 
\midrule
& & SC & BH &  GR & LR & Procedure 2 \\
\midrule
$\sigma_\alpha = 1.5\%$& FDX & 0.471 & 0.289 & 0.034 & 0.000 & 0.039\\ 
$\rho_\eta = 0$& FDR & 0.050 & 0.045 & 0.035 & 0.003 & 0.037\\ 
& Power (\%) & 92.9 & 92.5 & 91.0 & 68.4 & 91.3\\ 
\midrule 
$\sigma_\alpha = 1.25\%$& FDX & 0.493 & 0.340 & 0.032 & 0.003 & 0.049\\ 
$\rho_\eta = 0$& FDR & 0.051 & 0.046 & 0.028 & 0.002 & 0.032\\ 
& Power (\%) & 52.7 & 50.9 & 41.8 & 12.7 & 44.4\\ 
\midrule 
$\sigma_\alpha = 1.10\%$ & FDX & 0.477 & 0.382 & 0.045 & 0.045 & 0.041\\ 
$\rho_\eta = 0$& FDR & 0.050 & 0.045 & 0.006 & 0.006 & 0.018\\ 
& Power (\%) & 18.9 & 17.6 & 2.2 & 2.0 & 8.6\\ 
\midrule 
$\sigma_\alpha = 1.25\%$ & FDX & 0.500 & 0.341 & 0.031 & 0.000 & 0.041\\
$\rho_\eta = 0.1$ & FDR & 0.050 & 0.045 & 0.029 & 0.003 & 0.032\\
 & Power (\%) & 57.5 & 55.9 & 47.8 & 15.8 & 49.8\\
\midrule
$\sigma_\alpha = 1.25\%$ & FDX & 0.483 & 0.334 & 0.031 & 0.000 & 0.044\\
$\rho_\eta = 0.2$  & FDR & 0.050 & 0.046 & 0.030 & 0.003 & 0.033\\
& Power (\%) & 62.3 & 60.8 & 53.6 & 19.6 & 55.3\\
\midrule
\end{ctabular}
\end{boxedtable}

Table~\ref{Tbl.FinanceSimulation2} shows the performance of the same procedures under the realistic scenario where they must  learn the data-generating process from the data. In every scenario of $\sigma_\alpha$ and $\rho_\eta$, both Procedure 2 (\lfdr) and GR deliver FDX control above the desired threshold, with Procedure 2 possessing the higher FDX. This is likely because estimating $\pi$ from the simulated data is particularly challenging. A way to simplify the problem is to start with the assumption that the proportion of informative signal is zero, $\hat{\pi} = 0$. This leads to a substantial decrease in FDX, which nears the desired threshold, and comes at the cost of a slight power reduction.

\begin{boxedtable}{Stock return simulation (Data-Driven)}{Tbl.FinanceSimulation2}
\tablexplain{Across 5000 trading strategies the table compares the average FDX, FDR, and power obtained from Procedure 2 to those of  \cite{SunCai07} (SC), \cite{BenHoc95} (BH), \cite{GuoRomano07} (GR), and \cite{LehmannRomano05} (LR). In each simulation, the returns of 2000 stocks for 500 observations are generated from a linear factor structure model,  $ret_{it} = \alpha_i + \beta_i' F_t + \varepsilon_{it}$, where $\alpha_i$ represents the return that an investor could realize in excess of the risk generated by the factors, $F_t$ and is drawn from a $\mathcal{N}(0, \sigma_{\alpha}^2)$ and $\varepsilon_{it} \sim \mathcal{N}(0, \sigma_{\varepsilon}^2)$.  In each time period (i.e., month), we draw $S = 5000$ trading signals for each stock:~a fraction $\pi$ of the trading signals are informative (although imperfectly so) about the $\alpha$ of each stock: $s_{it} = \alpha_i + \eta_{it}$, where $\eta_{it} \sim \mathcal{N}(0, \sigma_{\eta})$.  A fraction $1-\pi$  contain just noise:~$s_{it} = \eta_{it}$. Informative and uninformative trading signals might share some common noise through the correlation coefficient $\rho_\eta = corr(\eta_i, \eta_j)$. Each month stocks are sorted into deciles based on the trading signal realization, and a long short portfolio is obtained from buying stocks in the top deciles and shorting stocks in the bottom deciles. A $t$-statistic of the portfolio regression alpha serves as the relevant test statistic in evaluating each trading strategy/long-short portfolio.  In the base case situation:~$\pi = 0.1$,  $\sigma_{\varepsilon} = 0.15$, and $\sigma_\eta = 0.2$. The table presents different simulations that vary $\sigma_\alpha \in \{1.1\%, 1.25\%, 1.5\% \}$ and $\rho_\eta \in \{0, 0.1, 0.2\}$. Each procedure is implemented in a Data-Driven fashion:~when necessary the econometrician estimates parameters directly from the observed data without knowledge of the true data-generating process. FDX control is implemented at $\gamma = 0.05$ and with a confidence level $1-\alpha = 0.95$, FDR control is implemented at a nominal level of $\alpha = 0.05$.}
\begin{ctabular}{l l rrrrrr} 
\midrule
&&&&&&\multicolumn{2}{c}{Procedure 2}\\
\cmidrule{7-8}
& & SC & BH &  GR & LR & $lfdr$ & $lfdr(\hat{\pi} = 0)$\\  
\midrule
$\sigma_\alpha = 1.5\%$ & FDX & 0.470 & 0.322 & 0.074 & 0.000 & 0.082 & 0.033\\ 
$\rho_\eta = 0$& FDR & 0.050 & 0.046 & 0.036 & 0.003 & 0.037 & 0.034\\ 
& Power (\%) & 92.9 & 92.5 & 91.0 & 68.4 & 91.3 & 90.6\\ 
\midrule 
$\sigma_\alpha = 1.25\%$ & FDX & 0.593 & 0.398 & 0.073 & 0.002 & 0.123 & 0.069\\ 
$\rho_\eta = 0$& FDR & 0.054 & 0.048 & 0.030 & 0.003 & 0.035 & 0.031\\ 
& Power (\%) & 53.6 & 51.6 & 42.7 & 13.2 & 45.5 & 43.4\\ 
\midrule 
$\sigma_\alpha = 1.1\%$ & FDX & 0.668 & 0.572 & 0.074 & 0.069 & 0.111 & 0.078\\ 
$\rho_\eta = 0$& FDR & 0.063 & 0.057 & 0.008 & 0.007 & 0.026 & 0.022\\ 
& Power (\%) & 22.1 & 20.4 & 2.8 & 2.3 & 11.1 & 9.6\\ 
\midrule 
$\sigma_\alpha = 1.25\%$ & FDX & 0.569 & 0.403 & 0.072 & 0.000 & 0.109 & 0.051\\
$\rho_\eta = 0.1$ & FDR & 0.053 & 0.047 & 0.030 & 0.003 & 0.035 & 0.031\\
& Power (\%) & 58.4 & 56.4 & 48.6 & 16.3 & 50.8 & 48.9\\
\midrule
$\sigma_\alpha = 1.25\%$  & FDX & 0.556 & 0.410 & 0.077 & 0.000 & 0.120 & 0.071\\
$\rho_\eta = 0.2$ & FDR & 0.053 & 0.048 & 0.032 & 0.003 & 0.035 & 0.032\\
& Power (\%) & 63.0 & 61.3 & 54.3 & 20.1 & 56.1 & 54.2\\
\midrule
\end{ctabular}
\end{boxedtable}

\end{document}